\documentclass[a4paper,12pt]{article}

\usepackage{amsmath,pgfplots,caption}
\usepackage{amssymb}
\usepackage{amsfonts}
\usepackage{graphicx,epstopdf}
\usepackage{mathrsfs}
\usepackage{color}
\usepackage{subfigure}

\newtheorem{theorem}{Theorem}
   
\newtheorem{proposition}[theorem]{Proposition}   
\newtheorem{lemma}[theorem]{Lemma}

\def\softO{\ensuremath{{O}{\,\tilde{ }\,}}}

\def\deg{{\rm deg}}
\def\bX{{\gamma}}
\def\bY{{\varphi}}

\usepackage[utf8]{inputenc}
\usepackage[T1]{fontenc}
\usepackage{lmodern}
\usepackage{comment}
\usepackage{amssymb,amsmath}
\usepackage{graphics}
\usepackage{amsfonts}
\usepackage{mathrsfs}
\usepackage{amscd}
\usepackage{color}
\usepackage{url}
\usepackage[plainpages=false,pdfpagelabels,colorlinks=true,citecolor=olive,hypertexnames=false]{hyperref}
\usepackage{alltt}

\usepackage{microtype}  
\usepackage{booktabs}   
\usepackage{array}      

\def\foorp{\hfill$\square$}

\newenvironment{proof}[1][]{\noindent {\bf Proof #1:\;}}{\hfill $\Box$}

\def\QQ{{\mathbb{Q}}} 
\def\RR{{\mathbb{R}}} 
\def\CC{{\mathbb{C}}} 

\def\f{\mathbf{f}}
\def\tildef{\tilde{\mathbf{f}}}
\def\g{\mathbf{g}}
\def\h{\mathbf{h}}
\def\q{\mathbf{q}}

\def\scC{{\mathscr{C}}}
\def\scS{{\mathscr{S}}}
\def\sfP{{\mathsf{P}}}
\def\sfQ{{\mathsf{Q}}}

\def\sing{{\rm sing}\,}
\def\crit{{\rm crit}\,}
\def\reg{{\rm reg}\,}
\def\Id{{\rm I}}
\def\minors{{\rm minors}\,} 

\newcommand{\zeroset}[1]{{\mathcal{Z}(#1)}}
\newcommand{\ideal}[1]{{I(#1)}}

\def\cc{{{C}}}
\def\H{{H}}

\def\lagrange{{\bell}}
\def\incidence{\mathcal{J}} 
\def\incidencereg{\mathcal{K}} 

\def\setC{{\mathcal{C}}}
\def\setV{{\incidence}}

\def\setZ{{\mathcal{Z}}}
\def\setH{{\mathcal{H}}}
\def\L{\mathcal{L}}

\def\zarH{{\mathscr{H}}}
\def\zarfiber{{\mathscr{A}}}
\def\zarU{{\mathscr{U}}}
\def\zarV{{\mathscr{V}}}
\def\zarW{{\mathscr{W}}}
\def\zarM{{\mathscr{M}}}

\def\sfG{{\mathsf{G}}} 
\def\sfH{{\mathsf{A}}} 

\def\M{{M}} 
\def\jac{{D\,}} 

\def\rank{{\rm rank}}
\def\deg{{\rm deg}}

\def\GL{{\mathrm{GL}}} 

\def\X{{x}} 
\def\Y{{y}} 
\def\Z{{z}} 

\def\T{{t}}
\def\frakh{\mathfrak{h}}

\def\matY{{Y}}

\def\x{{x}} 
\def\bfx{\mathbf{x}} 
\def\vecx{{\bf x}} 
\def\y{{y}} 

\def\vecy{{\bf y}}   
\def\z{{z}} 
 
\def\vecz{{\bf z}} 
\def\u{\mathbf{u}}  
\def\e{\mathbf{e}}  
\def\v{\mathbf{v}}  
\def\w{\mathbf{w}}  
\def\fiber{{\alpha}}

\def\bell{{\ell}}

\title{{Real root finding for rank defects in linear Hankel matrices}}

\author{Didier Henrion$^{1,2,3}$ \and Simone Naldi$^{1,2}$ \and Mohab Safey El Din$^{4,5,6,7}$}

\footnotetext[1]{CNRS, LAAS, 7 avenue du colonel Roche, F-31400 Toulouse; France.}
\footnotetext[2]{Universit\'e de Toulouse; LAAS, F-31400 Toulouse, France.}
\footnotetext[3]{Faculty of Electrical Engineering, Czech Technical University in Prague, Czech Republic.}
\footnotetext[4]{Sorbonne Universit\'es, UPMC Univ Paris 06, Equipe PolSys, LIP6, F-75005, Paris, France.}
\footnotetext[5]{INRIA Paris-Rocquencourt, PolSys Project, France.}
\footnotetext[6]{CNRS, UMR 7606, LIP6, France.}
\footnotetext[7]{Institut Universitaire de France.}

\date{\today}

\begin{document}

\maketitle

\thispagestyle{empty}

\begin{abstract}
 Let $H_0, \ldots, H_n$ be $m \times m$ 
matrices with entries in $\QQ$ and Hankel structure, i.e. constant skew diagonals.
We consider the linear Hankel matrix $H(\vecx)=H_0+\X_1H_1+\cdots+\X_nH_n$ and the problem of computing
      sample  points in  each connected component of the real algebraic set defined
 by the rank constraint ${\sf rank}(H(\vecx))\leq r$, 
 for a given integer $r \leq m-1$.
 Computing sample points in real algebraic sets defined by rank
 defects in linear matrices is a general problem that finds
      applications in many areas such as control theory, computational
 geometry, optimization, etc. Moreover, Hankel matrices appear in
 many  areas of engineering sciences. Also, since Hankel matrices are
 symmetric, any algorithmic development for this problem can be
 seen as a first step towards a dedicated exact algorithm for solving
 semi-definite programming problems, i.e. linear matrix
      inequalities. Under some genericity assumptions on the input (such as smoothness
 of an incidence variety), we design a probabilistic algorithm for
 tackling this problem. It is an adaptation of the so-called critical
 point method that takes advantage of the special structure of the
 problem.
 Its complexity reflects this: it is essentially quadratic
 in specific degree bounds on 
an incidence variety. We report on
 practical experiments and analyze how the algorithm takes advantage
 of this special structure. A first implementation outperforms
 existing implementations for computing sample points in general real
 algebraic sets: it tackles examples that are out of reach of the
 state-of-the-art.
\end{abstract}

\newpage 

\section{Introduction}

\paragraph*{Problem  statement and motivation}
Let $\QQ,\RR,\CC$ be respectively the
fields of rational, real and complex numbers, and let $m,n$ 
be positive integers. Given $m \times m$ 
      matrices $\H_0, \H_1, \ldots, \H_n$ with entries in $\QQ$ and Hankel
      structure, i.e. constant skew diagonals,
we  consider the {\it linear Hankel matrix} $\H(\vecx) = \H_0+\X_1\H_1+\ldots+\X_n\H_n$,
denoted $\H$ for short, and the algebraic set
\[
\setH_r =
      \{\vecx \in \CC^n : \rank \, \H(\vecx) \leq r\}.
\]
The goal of this paper is to provide an efficient algorithm for
computing at least one sample point per connected component of the
real algebraic set $\setH_r \cap \RR^n$.

Such an algorithm can be used to solve the matrix rank minimization
problem for $\H$. Matrix rank minimization mostly consists of
minimizing the rank of a given matrix whose entries are subject to
constraints defining a convex set. These problems arise in many
engineering or statistical modeling applications and have recently
received a lot of attention. Considering Hankel structures is
relevant since it arises in many applications (e.g. for model
reduction in linear dynamical systems described by Markov parameters, see
\cite[Section 1.3]{markovsky12}).


Moreover, an algorithm for computing sample points in each connected
component of $\setH_r\cap\RR^n$ can also be used to decide the
emptiness of the feasibility set $S=\{\vecx \in \RR^n : \H(\vecx)\succeq 0\}$.
Indeed, considering the minimum rank $r$ attained in the boundary of
$S$, it is easy to prove that one of the connected components of
$\setH_{r} \cap \RR^n$ is actually contained in $S$. Note also that
such feasibility sets, also called Hankel spectrahedra, have recently
attracted some attention (see e.g. \cite{BS14}).

The  intrinsic algebraic nature of our problem makes relevant the
design of
      exact algorithms to achieve reliability.
On the one hand, we aim at
exploiting   algorithmically the special Hankel structure to 
gain  efficiency. 
On the other hand, the design of a special algorithm
for the  case of linear Hankel matrices can bring the foundations of a
general    approach to e.g. the symmmetric case which is important for
semi-definite      programming, i.e. solving linear matrix inequalities.

\paragraph*{Related works and state-of-the-art} Our problem consists
of computing sample points in real algebraic sets. The first algorithm
for this problem is due to Tarski but its complexity was not
elementary recursive \cite{Tarski}. Next, Collins designed the
Cylindrical Algebraic Decomposition algorithm \cite{c-qe-1975}. Its
complexity is doubly exponential in the number of variables which is
far from being optimal since the number of connected components of a
real algebraic set defined by $n$-variate polynomial equations of
degree $\leq d$ is upper bounded by $O(d)^n$. Next, Grigoriev and
Vorobjov \cite{GV88} introduced the first algorithm based on critical
point computations computing sample point in real algebraic sets
within $d^{O(n)}$ arithmetic operations. This work has next been
improved and generalized (see \cite{BaPoRo06} and references therein)
from the complexity viewpoint. We may apply these algorithms to our
problem by computing all $(r+1)$-minors of 
the Hankel matrix and compute sample points in the real algebraic set
defined by the vanishing of these minors.  This is done in time
$(\binom{m}{r+1}\binom{n+r}{r})^{O(1)}+r^{O(n)}$ however since the
constant in the exponent is rather high, these algorithms did not lead
to efficient implementations in practice.  Hence, another series of
works, still using the critical point method but aiming at designing
algorithms that combine asymptotically optimal complexity and
practical efficiency has been developed (see e.g. \cite{BGHSS, SaSc03,
  GS14} and references therein).

Under regularity assumptions, these yield probabilistic algorithms
running in time which is essentially $O(d^{3n})$ in the smooth case
and $O(d^{4n})$ in the singular case (see \cite{S05}). Practically,
these algorithms are implemented in the library {\sc RAGlib} which
uses Gr\"obner bases computations (see \cite{faugere2012critical,
  Sp14} about the complexity of computing critical points with
Gr\"obner bases).

Observe that determinantal varieties such as $\setH_r$ are generically
singular (see \cite{bruns1988determinantal}). Also the
aforementioned algorithms do not exploit the structure of the
problem. In \cite{HNS2014}, we introduced an algorithm for computing
real points at which a 
{\em generic} linear square matrix of size $m$ has rank $\leq m-1$, by
exploiting the structure of the problem. However, because of the
requested genericity of the input linear matrix, we cannot use it for
linear Hankel matrices. Also, it does not allow to get sample points
for a given, smaller rank deficiency.

\paragraph*{Methodology      and main results} Our main result is an
algorithm      that computes sample points in each connected component of
$\setH_r      \cap \RR^n$ under some genericity assumptions on the entries
of the      linear Hankel matrix $\H$ (these genericity assumptions are
made      explicit below). Our algorithm exploits the Hankel structure of
the      problem. Essentially, its complexity is quadratic in a multilinear
B\'ezout bound on the number of complex solutions. Moreover, we find
that, heuristically, this bound is less than
${{m}\choose{r+1}}{{n+r}\choose{r}}{{n+m}\choose{r}}$.
Hence, for subfamilies of the real root finding problem on linear
Hankel matrices where the maximum rank allowed $r$ is fixed, the complexity
is essentially in $(nm)^{O(r)}$.

The very basic idea is to study the algebraic set $\setH_r\subset
\CC^n$ as the Zariski closure of the projection of an incidence
variety, lying in $\CC^{n+r+1}$. This variety encodes the fact that
the kernel of $\H$ has dimension $\geq m-r$.  This lifted variety
turns out to be generically smooth and equidimensional and defined by
quadratic polynomials with multilinear structure.  When these
regularity properties are satisfied, we prove that computing one point
per connected component of the incidence variety is sufficient to
solve the same problem for the variety $\setH_r \cap \RR^n$. We also
prove that these properties are generically satisfied. We remark that
this method is similar to the one used in \cite{HNS2014}, but in this
case it takes strong advantage of the Hankel structure of the linear
matrix, as detailed in Section \ref{sec:prelim}. This also reflects on
the complexity of the algorithm and on practical performances.

Let $\cc$ be a connected component of $\setH_r\cap\RR^n$, and and
$\Pi_1,\pi_1$ be the canonical projections $\Pi_1: (\x_1, \ldots,
\x_n, \y_1, \ldots, \y_{r+1})\to \x_1$ and $\pi_1: (\x_1, \ldots,
\x_n)\to \x_1$. We prove that in generic coordinates, either {\em (i)}
$\pi_1(C) = \RR$ or {\em (ii)} there exists a critical point of the
restriction of $\Pi_1$ to the considered incidence variety. Hence,
after a generic linear change of variables, the algorithm consists of
two main steps: {\em (i)} compute the critical points of the
restriction of $\Pi_1$ to the incidence variety and {\em (ii)}
instantiating the first variable $\X_1$ to a generic value and perform
a recursive call following a geometric pattern introduced in
\cite{SaSc03}.

This latter step ({\em i}) is actually performed by building the
Lagrange system associated to the optimization problem whose solutions
are the critical points of the restriction of $\pi_1$ to the incidence
variety. Hence, we use the algorithm in \cite{jeronimo2009deformation}
to solve it.  One also observes heuristically that these Lagrange
systems are typically zero-dimensional.

However, we were not able to prove this finiteness property,
but we
      prove that it holds when we restrict the optimization step
to the set
      of points $\vecx \in \setH_r$ such that $\rank \, \H(\vecx)
= p$, for
      any $0 \leq p \leq r$. However, this is sufficient 
to
      conclude that there are finitely many critical points of the
restriction
      of $\pi_1$ to $\setH_r \cap \RR^n$, and that the
algorithm
      returns the output correctly.

When the
      Lagrange system has dimension $0$, the complexity of
solving
      its equations is essentially
quadratic
      in the number of its complex solutions.
As previously announced, by the
      structure of these systems one can deduce multilinear
B\'ezout bounds on
      the number of solutions that are polynomial in $nm$ when
$r$ is fixed, and polynomial in $n$ when $m$ is fixed. 
This complexity result outperforms
the
      state-of-the-art algorithms.
We finally
      remark that the complexity gain is reflected also
in the
      first implementation of the algorithm, which allows
to solve instances of our problem that are out of reach of the
general algorithms implemented in {\sc RAGlib}.

\paragraph*{Structure
      of the paper}
The paper
      is structured as follows. Section \ref{sec:prelim} contains
preliminaries
      about Hankel matrices and the basic notation of the
paper; we
      also prove that our regularity assumptions are generic. In
Section
      \ref{sec:algo} we describe the algorithm and prove its
correctness.
      This is done by using preliminary results proved in
Sections
      \ref{sec:dimension} and \ref{sec:closedness}. Section
      \ref{ssec:algo:complexity} contains the complexity analysis
and bounds
      for the number of complex solutions of the output of the
algorithm.
Finally, Section \ref{sec:exper} presents the results of our
experiments on generic linear Hankel matrices, and comparisons
with the state-of-the-art algorithms for the real root finding
problem.

\section{Notation and preliminaries} \label{sec:prelim}

\paragraph*{Basic notations} \label{ssec:prelim:basic}

We denote by $\GL(n, \QQ)$ (resp. $\GL(n, \CC)$)
the set of $n \times n$ non-singular matrices with rational (resp.
complex) entries. For a matrix $M \in \CC^{m \times m}$ and an integer
$p \leq m$, one denotes with $\minors(p,M)$ the list of determinants
of $p \times p$ sub-matrices of $M$. We denote by $M'$ the
transpose matrix of $M$.

Let $\QQ[\vecx]$ be the ring of polynomials on $n$ variables $\vecx = (\X_1, \ldots, \X_n)$
and let $\f = (f_1, \ldots, f_p) \in \QQ[\vecx]^p$ be a polynomial system.
The common zero locus of the entries of $\f$ is denoted by
$\zeroset{\f} \subset \CC^n$, and its dimension with $\dim \, \zeroset{\f}$. The ideal generated by $\f$ is denoted by
$\left\langle \f \right\rangle$, while if $\mathcal{V} \subset \CC^n$ is any set, the
ideal of polynomials vanishing on $\mathcal{V}$ is denoted by $\ideal{\mathcal{V}}$, while the
set of regular (resp. singular) points of $\mathcal{V}$ is denoted by $\reg \, \mathcal{V}$
(resp. $\sing \, \mathcal{V}$). If $\f = (f_1, \ldots, f_p) \subset \QQ[\vecx]$, we denote
by $\jac \f = \left( \partial f_i / \partial \X_j \right)$ the Jacobian matrix
of $\f$. We denote by $\reg(\f) \subset \zeroset{\f}$ the subset where
$\jac \f$ has maximal rank.

A set $\mathcal{E} \subset \CC^n$ is locally closed if $\mathcal{E} = \setZ
\cap \mathscr{O}$ where $\setZ$ is a Zariski closed set and $\mathscr{O}$ is a
Zariski open set. 

Let $\mathcal{V} = \zeroset{\f} \subset \CC^n$ be a smooth equidimensional algebraic set, of dimension $d$,
and let $\g \colon \CC^n \to \CC^p$ be an algebraic map. The set of critical points of
the restriction of $\g$ to $\mathcal{V}$ is the solution set of $\f$ and of the $(n-d+p)-$minors
of the matrix $\jac (\f, \g)$, and it is denoted by $\crit(\g, \mathcal{V})$. Finally, if $\mathcal{E}
\subset \mathcal{V}$ is a locally closed subset of $\mathcal{V}$, we denote by $\crit(\g, \mathcal{E})
= \mathcal{E} \cap \crit(\g, \mathcal{V})$.

Finally, for $M \in \GL(n, \CC)$ and $f \in \QQ[\vecx]$, we denote by $f^M(\vecx) = f(M\,\vecx)$,
and if $\f=(f_1, \ldots, f_p) \subset \QQ[\vecx]$ and $\mathcal{V} = \zeroset{\f}$, by
$\mathcal{V}^M = \zeroset{\f^M}$ where $\f^M = (f^M_1, \ldots, f^M_p)$.

\paragraph*{Hankel structure} \label{ssec:prelim:hanktoep}
Let $\{h_1, \ldots, h_{2m-1}\} \subset \QQ$. The matrix $\H = (h_{i+j-1})_{1 \leq i,j \leq m}
\in \QQ^{m \times m}$ is called a Hankel matrix, 
and we use the notation $\H = {\sf Hankel}(h_1, \ldots, h_{2m-1})$.
The structure of a Hankel matrix induces structure on its kernel. By
\cite[Theorem 5.1]{heinig1984algebraic}, one has that if $\H$ is a Hankel
matrix of rank at most $r$, then there exists a non-zero vector $\vecy = (\y_1,
\ldots, \y_{r+1}) \in \QQ^{r+1}$ such that the columns of the $m \times
(m-r)$ matrix
\[
\matY(\vecy)=
\begin{bmatrix}
\vecy      & 0      & \ldots & 0      \\
0      & \vecy      & \ddots & \vdots \\
\vdots & \ddots & \ddots & 0      \\
0      & \ldots & 0      & \vecy      \\
\end{bmatrix}
\]
generate a $(m-r)-$dimensional subspace of the kernel of $\H$. 
We observe that $\H \, \matY(\vecy)$ is also a Hankel matrix.

The product $\H \, \matY(\vecy)$ can be re-written as a matrix-vector product
$\tilde{\H} \, y$, with $\tilde{\H}$ a given rectangular Hankel matrix.
Indeed, let $\H={\sf Hankel}(h_1, \ldots, h_{2m-1})$. Then, as previously
observed, $\H \, \matY(\vecy)$ is a rectangular Hankel matrix, of size
$m \times (m-r)$, whose entries coincide with the entries of
\[
\tilde{\H} \, \vecy =
\begin{bmatrix}
h_1      & \ldots & h_{r+1} \\
\vdots   &        & \vdots \\
h_{2m-r-1} & \ldots & h_{2m-1}
\end{bmatrix}
\begin{bmatrix}
\y_1    \\
\vdots \\
\y_{r+1}
\end{bmatrix}.
\]

Let $H(\vecx)$ be a linear Hankel matrix.
From \cite[Corollary 2.2]{conca1998straightening} we deduce that, for $p \leq r$, then
the ideals $\left\langle \minors(p+1,\H(\vecx)) \right\rangle$ and
$\left\langle\minors(p+1,\tilde{\H}(\vecx))\right\rangle$ coincides. One deduces that
$\vecx = (\X_1, \ldots, \X_n) \in \CC^n$ satisfies $\rank \, \H(\vecx) = p$
if and only if it satisfies $\rank \, \tilde{\H}(\vecx) = p$.

\paragraph*{Basic sets} \label{ssec:prelim:polsys}
We first recall that the linear matrix $\H(\vecx) = \H_0 + \X_1\H_1 + \ldots + \X_n\H_n$,
where each $H_i$ is a Hankel matrix, is also a Hankel matrix. It is identified
by the $(2m-1)(n+1)$ entries of the matrices $\H_i$. Hence we often consider $\H$ as an element of
$\CC^{(2m-1)(n+1)}$. For $M \in \GL(n, \QQ)$, we denote by $\H^M(\vecx)$ the linear
matrix $\H(M \vecx)$.

We define in the following the main algebraic sets appearing during
the execution of our algorithm, given $\H \in \CC^{(2m-1)(n+1)}$,
$0 \leq p \leq r$, $M \in \GL(n, \CC)$ and $\u = (u_1, \ldots, u_{p+1}) \in \QQ^{p+1}$.


{\it Incidence varieties.} We consider the polynomial system
\[
\begin{array}{lccl}
  \f(\H^M,\u,p): &  \CC^{n} \times \CC^{p+1} & \longrightarrow & \CC^{2m-p-1} \times \CC \\
  &  (\vecx, \vecy) & \longmapsto & 
  \left ((\tilde{H}(M \, \vecx)\, \vecy)', \u'\vecy-1 \right )
\end{array}
\]
where $\tilde{H}$ has been defined in the previous section.
We denote by $\setV(\H^M, \u, p) = \zeroset{\f_p(\H^M,\u)} \subset \CC^{n+p+1}$ and simply
$\setV=\setV(\H^M,\u,p)$ and $\f=\f(\H^M, \u, p)$ when $p, \H, M$ and $\u$ are clear.
We also denote by $\incidencereg(\H^M, \u, p) = \incidence(\H^M,\u, p) \cap
\{(\vecx, \vecy)\in \CC^{n+p+1} : \rank \,\H(\vecx)=p\}.$

{\it Fibers.} Let $\fiber \in \QQ$. We denote by
$\f_{\fiber}(\H^M, \u, p)$ (or simply $\f_{\fiber}$) the polynomial system obtained
by adding $\X_1-\fiber$ to $\f(\H^M, \u, p)$. The resulting algebraic set
$\zeroset{\f_{\fiber}}$, denoted by $\setV_{\fiber}$, equals $\setV \cap \zeroset{\X_1-\fiber}$.

{\it Lagrange systems.}
Let $\v \in \QQ^{2m-p}$. Let $\jac_1\f$ denote the matrix of size $c \times (n+p)$
obtained by removing the first column of $\jac \f$ (the derivative
w.r.t. $\X_1$), and define $\bell=\bell(\H^M, \u, \v, p)$ as the map
\[
\begin{array}{lrcl}
\bell : & \CC^{n+2m+1} & \to     & \CC^{n+2m+1} \\
                        & (\vecx,\vecy,\vecz)    & \mapsto & (\tilde{\H}(M \, \vecx) \, \vecy, \u'\vecy-1, \vecz'\jac_1\f, \v'\vecz-1)
\end{array}
\]
where $\vecz=(\z_1, \ldots, \z_{2m-p})$ stand for Lagrange multipliers. We
finally define $\setZ(\H^M, \u, \v, p) = \zeroset{\bell(\H^M, \u, \v, p)} \subset
\CC^{n+2m+1}.$

\paragraph*{Regularity property $\sfG$} \label{ssec:prelim:regul}

We say that a polynomial system $\f \in \QQ[x]^c$ satisfies Property $\sfG$ if
the Jacobian matrix $\jac \, \f$ has maximal rank at any point of $\zeroset{\f}$.
We remark that this implies that:
\begin{enumerate}
\item the ideal $\ideal{\f}$ is radical;
\item the set $\zeroset{\f}$ is either empty or smooth and equidimensional of co-dimension $c$.
\end{enumerate}

We say that $\lagrange(\H^M, \u, \v, p)$ satisfies $\sfG$ over
$\incidencereg(H^M, \u, p)$ if the following holds: for $(\vecx,\vecy,\vecz) \in
\setZ(\H^M, \u, \v, p)$ such that $(\vecx,\vecy) \in \incidencereg(H^M, \u, p)$,
the matrix $\jac(\lagrange(\H^M, \u, \v, p))$ has maximal rank at $(\vecx,\vecy,\vecz)$.

Let $\u  \in \QQ^{p+1}$. We say that $\H \in \CC^{(2m-1)(n+1)}$ satisfies Property
$\sfG$ if $\f(\H, \u, p)$ satisfies Property $\sfG$ for all $0 \leq p \leq r$.

The first result essentially shows that $\sfG$ holds
for $\f(\H^M, \u, p)$ (resp. $\f_\fiber(\H^M, \u, p)$) when the
input parameter $\H$ (resp. $\fiber$) is generic enough.


\begin{proposition} \label{prop:regularity}
Let $M \in \GL(n,\CC)$.
\begin{itemize}
\item[(a)] There exists a non-empty Zariski-open set $\zarH \subset
  \CC^{(2m-1)(n+1)}$ such that, if $\H \in \zarH \cap
  \QQ^{(2m-1)(n+1)}$, for all $0 \leq p \leq r$ and $\u \in \QQ^{p+1}-\{\mathbf{0}\}$,
  $\f(\H^M, \u, p)$ satisfies Property
  $\sfG$;
\item[(b)] for $\H \in \zarH$, and $0 \leq p \leq r$, if $\setV(\H^M, \u, p) \neq \emptyset$
  then $\dim \, \setH_p \leq n-2m+2p+1$; 
\item[(c)] For $0 \leq p \leq r$ and $\u \in \QQ^{p+1}$, if $\f(\H^M, \u, p)$ satisfies
  $\sfG$, there exists a non-empty Zariski open set $\zarfiber \subset \CC$ such
  that, if $\fiber \in \zarfiber$, the polynomial system $\f_{\fiber}$
  satisfies $\sfG$;
\end{itemize}
\end{proposition}

\proof
Without loss of generality, we can assume that $M = \Id_n$. We let $0
\leq p \leq r$, $\u \in \QQ^{p+1}-\{\mathbf{0}\}$ and recall that we identify the
space of linear Hankel matrices with $\CC^{(2m-1)(n+1)}$. This space is
endowed by the variables $\frakh_{k,\ell}$ with $1\leq k \leq 2m-1$
and $0\leq \ell \leq n$; the generic linear Hankel matrix is then
given by
$\mathfrak{H}=\mathfrak{H}_0+\X_1\mathfrak{H}_1+\cdots+\X_n\mathfrak{H}_n$
with $\mathfrak{H}_i={\sf Hankel}(\frakh_{1,i}, \ldots, \frakh_{2m-1, i})$.

We consider the map
\[
  \begin{array}{lccc}
  q : &  \CC^{n+(p+1)+(2m-1)(n+1)} & \longrightarrow & \CC^{2m-p} \\
      &  (\vecx, \vecy, \H) & \longmapsto& \f(\H, \u, p)
  \end{array}
\]
and, for a given $\H \in \CC^{(2m-1)(n+1)}$, its section-map $q_\H
\colon \CC^{n+(p+1)} \to \CC^{2m-p}$ sending $(\vecx,\vecy)$ to
$q(\vecx,\vecy,\H)$. We also consider the map $\tilde{q}$ which associates
to $(\vecx, \vecy, \H)$ the entries of $\tilde{H}\vecy$ and its section map
$\tilde{q}_H$; we will consider these latter maps over the open set
$O=\{(\vecx, \vecy)\in \CC^{n+p+1}\mid \vecy\neq \mathbf{0}\}$. We prove below
that $\mathbf{0}$ is a regular value for both $q_H$ and $\tilde{q}_H$.

Suppose first that $q^{-1}(\mathbf{0}) = \emptyset$
(resp. $\tilde{q}^{-1}(\mathbf{0})$). We deduce that for all $\H \in
\CC^{(2m-1)(n+1)}$, $q_H^{-1}(\mathbf{0}) = \emptyset$ (resp
$\tilde{q}_H^{-1}(\mathbf{0}) = \emptyset$) and $\mathbf{0}$ is a
regular value for both maps $q_H$ and $\tilde{q}_H$. Note also that
taking $\zarH = \CC^{(2m-1)(n+1)}$, we deduce that $\f(\H, \u, p)$
satisfies $\sfG$.

Now, suppose that $q^{-1}(\mathbf{0})$ is not empty and let
$(\vecx,\vecy,\H) \in q^{-1}(0)$.  Consider the Jacobian matrix $\jac q$ of
the map $q$ with respect to the variables $\vecx,\vecy$ and the entries of $\H$,
evaluated at $(\vecx,\vecy,\H)$. We consider the submatrix of $\jac q$
by selecting the column corresponding to:
\begin{itemize}
\item the partial derivatives with respect to $\frakh_{1, 0}, \ldots,
  \frakh_{2m-1, 0}$;
\item the partial derivatives with respect to $\Y_1, \ldots,
  \Y_{p+1}$.
\end{itemize}
We obtain a $(2m-p) \times (2m+p)$ submatrix of $\jac q$; we prove below
that it has full rank $2m-p$.

Indeed, remark that the $2m-p-1$ first lines correspond to the entries
of $\tilde{\H}\vecy$ and last line corresponds to the derivatives of
$\u'\vecy-1$. Hence, the structure of this submatrix is as below
\[
\begin{bmatrix}
  \y_1 &  \ldots & \y_{p+1} & 0        & \ldots  &  0 &0& \cdots & 0 \\
  0   & \y_1 & \ldots & \y_{p+1} & \ldots  &  0  & & \\
  \vdots & & \ddots & & \ddots & & \vdots & & \vdots\\
  \vdots & &        &\y_1 & \ldots & \y_{p+1}& 0 &   & 0\\
  0      & & \cdots & & \cdots &  0 & u_1 & \cdots & u_{p+1}\\
\end{bmatrix}
\]
Since this matrix is evaluated at the solution set of $\u'\vecy-1=0$,
we deduce straightforwardly that one entry of $\u$ and one entry of $\vecy$
are non-zero and that the above matrix is full rank and that
$\mathbf{0}$ is a regular value of the map $q$.

We can do the same for $\jac \tilde{q}$ except the fact that we do not
consider the partial derivatives with respect to $\Y_1, \ldots,
\Y_{p+1}$. The $(2m-p-1) \times (2m-1)$ submatrix we obtain corresponds
to the upper left block containing the entries of $\vecy$. Since
$\tilde{q}$ is defined over the open set $O$ in which $\vecy\neq
\mathbf{0}$, we also deduce that this submatrix has full rank
$2m-p-1$.

By Thom's Weak Transversality Theorem one deduces that there exists a
non-empty Zariski open set $\zarH_p \subset \CC^{(2m-1)(n+1)}$ such
that if $\H \in \zarH_p$, then $\mathbf{0}$ is a regular value of
$q_\H$ (resp. $\tilde{q}_\H$). We deduce that for $\H \in \zarH_p$, the
polynomial system $\f(\H, \u, p)$ satisfies $\sfG$ and using the
Jacobian criterion \cite[Theorem 16.19]{Eisenbud95},
$\incidence(\H, \u, p)$ is either empty or smooth equidimensional of
dimension $n-2m+2p+1$. This proves assertion (a), with $\zarH =
\bigcap_{0 \leq p \leq r}\zarH_p$.

Similarly, we deduce that $\tilde{q}_\H^{-1}(\mathbf{0})$ is either
empty or smooth and equidimensional of dimension $n-2m+2p+2$.  Let
$\Pi_\vecx$ be the canonical projection $(\vecx, \vecy)\to \vecx$;
note that for
any $\vecx \in \setH_r$, the dimension of $\Pi_\vecx^{-1}(x)\cap
\tilde{q}_H^{-1}(\mathbf{0})$ is $\geq 1$ (by homogeneity
of the $\vecy$-variables). By the Theorem on the Dimension of Fibers
\cite[Sect.6.3,Theorem 7]{Shafarevich77}, we deduce that
$n-2m+2p+2-\dim(\setH_p) \geq 1$. We deduce that for
$\H \in \zarH$, $\dim(\setH_p) \leq n-2m+2p+1$ which proves assertion
(b).

It remains to prove assertion (c). We assume that 
$\f(\H, \u, p)$ satisfies
$\sfG$. Consider the restriction of the map $\Pi_1 \colon \CC^{n+p+1}
\to \CC$, $\Pi_1(\vecx,\vecy)=\x_1$, to $\setV(\H, \u, p)$, which is smooth
and equidimensional by assertion (a).

By Sard's Lemma \cite[Section 4.2]{SaSc13}, the set of critical values
of the restriction of $\Pi_1$ to $\setV(\H, \u, p)$ is finite. Hence,
its complement $\zarfiber \subset \CC$ is a non-empty Zariski open
set. We deduce that for $\fiber \in \zarfiber$, the Jacobian matrix of
$\f_\alpha(\H, \u, p)$ satisfies $\sfG$. 
\foorp

%
%
%
%
%


\section{Algorithm and correctness} \label{sec:algo}

In this section we present the algorithm, which is called {\sf LowRank\-Hankel},
and prove its correctness.

\subsection{Description} \label{ssec:algo:desc}

\paragraph*{Data representation}

The algorithm takes as {\it input} a couple $(\H, r)$, where $\H = (\H_0, \H_1, \ldots,
\H_n)$ encodes $m \times m$ Hankel matrices with entries in $\QQ$, defining the
linear matrix $\H(\vecx)$, and $0 \leq r \leq m-1$.

The {\it output} is represented by a rational parametrization, that is a polynomial system
\[
\q = (q_0(\T), q_1(\T), \ldots, q_n(\T), q(\T)) \subset \QQ[\T]
\]
of univariate polynomials, with $gcd(q,q_0)=1$. The set of solutions of
\[
\X_i-q_i(\T)/q_0(\T) = 0, \,\,i=1 \ldots n \qquad q(\T)=0
\]
is clearly finite and expected to contain at least one point per connected component
of the algebraic set $\setH_r \cap \RR^n$.

\paragraph*{Main subroutines and formal description}

We start by describing the main subroutines we use. 


\noindent {\sf ZeroDimSolve}. It takes as input a polynomial system
defining an algebraic set $\setZ\subset \CC^{n+k}$ and a subset of
variables $\vecx=(\X_1, \ldots, \X_n)$. 
If $\setZ$ is
finite, it returns a rational parametrization of the projection of
$\setZ$ on the $\vecx$-space else it returns an empty list.

\noindent {\sf ZeroDimSolveMaxRank}. It takes as input a polynomial
system $\f=(f_1, \ldots, f_c)$ such that $Z=\{\vecx \in
\CC^{n+k} \mid {\sf rank}(\jac\f(\vecx))=c\}$ is finite and a
subset of variables $\vecx=(\X_1, \ldots, \X_n)$ that endows $\CC^n$.  It
returns {\sf fail: the assumptions are not satisfied} if assumptions
are not satisfied, else it returns a rational parametrization of the
projection of $\setZ$ on the $\vecx$-space.

\noindent {\sf Lift}. It takes as input a rational parametrization of
a finite set $\setZ \subset \CC^N$ and a number $\fiber \in \CC$, and
it returns a rational parametrization of $\{(\fiber, \bfx) \, : \,
\bfx \in \setZ\}$. 

\noindent {\sf Union}. It takes as input two rational parametrizations
encoding finite sets $\setZ_1, \setZ_2$ and it returns a rational
parametrization of $\setZ_1 \cup \setZ_2$.

\noindent {\sf ChangeVariables}. It takes as input a rational
parametrization of a finite set $\setZ \subset \CC^N$ and a
non-singular matrix $M \in \GL(N, \CC)$. It returns a rational
parametrization of $\setZ^M$. 



The algorithm {\sf LowRankHankel} is recursive, and it
assumes that its input $\H$ satisfies Property $\sfG$.





${\sf LowRankHankel}(\H,r)$:
\begin{enumerate}
\item \label{step:rec:1} If $n < 2m-2r-1$ then return $[\,]$. 
\item\label{step:rec:choice1} Choose randomly $M\in \GL(n, \QQ)$,  $\fiber \in \QQ$ and
  $\u_p \in \QQ^{p+1}$, $\v_p \in \QQ^{2m-p}$ for $0\leq p \leq r$. 
\item \label{step:rec:2} If $n = 2m-2r-1$ then return ${\sf ZeroDimSolve}(\f(\H,\u_r, r), \vecx))$.
\item Let $\sfP={\sf ZeroDimSolve}(\lagrange(\f(\H,\u_r,r), \v))$
\item \label{step:rec:3} If $\sfP=[]$ then for $p$ from 0 to $r$ do
  \begin{enumerate}
    \item $\sfP'={\sf ZeroDimSolveMaxRank}(\lagrange(\H^M, \u_p, \v_p), \vecx)$;
    \item $\sfP = {\sf Union}(\sfP, \sfP')$
  \end{enumerate}
\item \label{step:rec:5} ${\sf Q}={\sf Lift}({\sf LowRankHankel}({\sf Subs}(\X_1=\fiber, \H^M),r), \fiber)$;
\item \label{step:rec:6} return({\sf ChangeVariables}({\sf Union}($\sfQ, \sfP$), $M^{-1}$)).
\end{enumerate}

\subsection{Correctness} \label{ssec:algo:corr} \label{sssec:algo:prelimresult}


The correctness proof is based on the two following results that are
proved in Sections \ref{sec:dimension} and \ref{sec:closedness}.

The first result states that when the input matrix $H$ satisfies
$\sfG$ and that, for a generic choice of $M$ and $\v$, and for all
$0 \leq p \leq r$, the set of solutions $(\vecx, \vecy, \vecz)$ to $\bell(\H^M, \u,
\v, p)$ at which $\rank\,\tilde{\H}(x)=p$ is finite and contains
$\crit(\pi_1, \incidencereg(\H^M,\u, p))$.
\begin{proposition} \label{prop:dimension}\label{PROP:DIMENSION} Let
  $\zarH$ be the set defined in Proposition \ref{prop:regularity} and
  let $\H \in \zarH$ and $\u \in \QQ^{p+1}-\{\mathbf{0}\}$ for $0 \leq p
  \leq r$. There exist non-empty Zariski open sets $\zarM_1 \subset
  \GL(n,\CC)$ and $\zarV \subset \CC^{2m-p}$ such that if $M \in
  \zarM_1 \cap \QQ^{n \times n}$ and $\v \in \zarV \cap \QQ^{2m-p}$,
  the following holds:
\begin{itemize}
\item[(a)] $\lagrange(\H^M, \u, \v, p)$ satisfies $\sfG$ over
  $\incidencereg(H^M, u, p)$;
\item[(b)] the projection of $\reg(\lagrange(\H^M, \u, \v, p))$ on the
  $(\vecx,\vecy)$-spa\-ce contains $\crit(\Pi_1, \incidencereg(\H^M, \u,
  p))$
\end{itemize}
\end{proposition}



\begin{proposition}\label{prop:closedness}\label{PROP:CLOSEDNESS}
  Let $\H \in \zarH$, $0\leq p\leq r$ and $d_p = n-2m+2p+1$ and
  $\setC$ be a connected component of $\setH_p \cap \RR^n$.  Then
  there exist non-empty Zariski open sets $\zarM_2 \subset \GL(n,\CC)$
  and $\zarU \subset \CC^{p+1}$ such that for any $M \in \zarM_2 \cap
  \QQ^{n \times n}$, $\u \in \zarU \cap \QQ^{p+1}$, the following
  holds:
  \begin{itemize}
  \item[(a)] for $i = 1, \ldots, d_p$, $\pi_i(\setC^M)$ is closed;
  \item[(b)] for any $\fiber \in \RR$ in the boundary of
    $\pi_1(\setC^M)$,  $\pi_1^{-1}(\fiber) \cap \setC^M$
    is finite;
  \item[(c)] for any $\vecx\in \pi_1^{-1}(\fiber) \cap \setC^M$ and $p$
    such that $\rank \, \tilde{\H}_p(\vecx)=p$, there exists $(\vecx, \vecy) \in
    \RR^n \times \RR^{p+1}$ such that $(\vecx,\vecy) \in \setV (\H^\M, \u,
    p)$.
  \end{itemize}
\end{proposition}



Our algorithm is probabilistic and its correctness depends on the
validity of the choices that are made at Step
\ref{step:rec:choice1}. We make this assumption that we formalize
below.


We need to distinguish the choices of $M, \u$ and $\v$ that are made
in the different calls of {\sf LowRankHankel}; each of these parameter
must lie in a non-empty Zariski open set defined in Propositions
\ref{prop:regularity}, \ref{prop:dimension} and \ref{prop:closedness}.

We assume that the input matrix $\H$ satisfies $\sfG$; we denote it by
$\H^{(0)}$, where the super script indicates that no recursive call has been
made on this input; similarly $\fiber^{(0)}$ denotes the choice of
$\fiber$ made at Step \ref{step:rec:choice1} on input
$\H^{(0)}$. Next, we denote by $\H^{(i)}$ the input of {\sf
  LowRankHankel} at the $i$-th recursive call and by
$\zarfiber^{(i)}\subset \CC$ the non-empty Zariski open set defined in
Proposition \ref{prop:regularity} applied to $H^{(i)}$. Note that if
$\fiber^{(i)}\in \zarfiber^{(i)}$, we can deduce that $\H^{(i+1)}$
satisfies $\sfG$.

Now, we denote by $\zarM_1^{(i)},\zarM_2^{(i)}$ and
$\zarU^{(p,i)},\zarV^{(p,i)}$ the open sets defined in Propositions
\ref{prop:regularity}, \ref{prop:dimension} and \ref{prop:closedness}
applied to $H^{(i)}$, for $0 \leq p \leq r$ and where $i$ is the depth
of the recursion.

Finally, we denote by $M^{(i)} \in \GL(n, \QQ)$, $\u^{(i)}_p\in
\QQ^{p+1}$ and $\v^{(i)}_p$, for $0 \leq p \leq r$, respectively the
matrix and the vectors chosen at Step \ref{step:rec:choice1} of the
$i$-th call of ${\sf LowRankHankel}$.

\noindent {\bf Assumption $\sfH$}. We say that $\sfH$
is satisfied if $M^{(i)}$, $\fiber^{(i)}$, $\u^{(i)}_p$ and $\v^{(i)}_p$ 
satisfy:
\begin{itemize}
\item $M^{(i)} \in  (\zarM_1^{(i)} \cap
  \zarM_2^{(i)}) \cap \QQ^{i \times i}$; 
\item $\fiber^{(i)} \in \zarfiber^{(i)}$. 
\item $\u^{(i)}_p \in \zarU^{(p,i)} \cap
  \QQ^{p+1}-\{\mathbf{0}\}$, for $0 \leq p \leq r$;
\item $\v_p^{(i)} \in \zarV^{(p,i)} \cap \QQ^{2m-p}-\{\mathbf{0}\}$ for  $0 \leq p \leq r$;
\end{itemize}


\begin{theorem}
  Let $\H$ satisfy $\sfG$. Then, if
  $\sfH$ is satisfied, ${\sf LowRankHankel}$ with input $(\H, r)$, returns a
  rational para\-met\-rization that encodes a finite algebraic set in $\setH_r$ meeting
  each connected component of $\setH_r \cap \RR^n$.
\end{theorem}
\begin{proof}
The proof is by decreasing induction on the depth of the recursion.

When $n<2m-2r-1$, $\setH_r$ is empty since the input $\H$ satisfies
$\sfG$ (since $\sfH$ is satisfied). In this case, the output defines the empty set.



When $n=2m-2r-1$, since $\sfH$ is satisfied, by Proposition \ref{PROP:REGULARITY},
either $\setH_r = \emptyset$ or $\dim\,\setH_r = 0$. Suppose $\setH_r = \emptyset$.
Hence $\setV_r = \emptyset$, since the projection of $\setV_r$ on
the $\vecx-$space is included in $\setH_r$. Suppose now that $\dim \setH_r
= 0$: Proposition \ref{prop:closedness} guarantees that the output
of the algorithm defines a finite set containing $\setH_r$.

Now, we assume that $n>2m-2r-1$; our induction assumption is that for
any $i\geq 1$ ${\sf LowRankHankel}(\H^{(i)}, r)$ returns a rational
parametrization that encodes a finite set of points in the algebraic
set defined by ${\sf rank}(\H^{(i)})\leq r$ and that meets every
connected component of its real trace.

Let $\cc$ be a connected component of $\setH_r \cap \RR^n$.  To keep
notations simple, we denote by $M \in \GL(n, \QQ)$, $\u_p$ and $\v_p$
the matrix and vectors chosen at Step \ref{step:rec:choice1} for
$0\leq p \leq r$. Since $\sfH$ holds one can apply Proposition
\ref{prop:closedness}. We deduce that the image $\pi_1(\cc^M)$ is
closed. Then, either $\pi_1(\cc^M) = \RR$ or it is a closed interval.

Suppose first that $\pi_1(\setC^M) = \RR$. Then for $\fiber \in \QQ$
chosen at Step \ref{step:rec:choice1}, $\pi_1^{-1}(\fiber) \cap
\setC^M \neq 0$. Remark that $\pi_1^{-1}(\fiber) \cap \cc^M$ is the
union of some connected components of $\setH^{(1)}_r \cap \RR^{n-1} =
\{\vecx=(\x_2, \ldots, \x_n) \in \RR^{n-1} : \rank \, \H^{(1)} (\vecx) \leq
r\}$.  Since $\sfH$ holds, assertion (c) of Proposition
\ref{prop:regularity} implies that $\H^{(1)}$ satisfies $\sfG$. We
deduce by the induction assumption that the parametrization returned
by Step \ref{step:rec:5} where ${\sf LowRankHankel}$ is called
recursively defines a finite set of points that is contained in
$\setH_r$ and that meets $\cc$.

Suppose now that $\pi_1(\cc^M) \neq \RR$. By Proposition
\ref{prop:closedness}, $\pi_1(\cc^M)$ is closed. Since $\cc^M$ is
connected, $\pi_1(\cc^M)$ is a connected interval, and since
$\pi_1(\cc^M) \neq \RR$ there exists $\beta$ in the boundary of
$\pi_1(\cc^M)$ such that $\pi_1(\cc^M) \subset [\beta, +\infty)$ or
$\pi_1(\cc^M) \subset (-\infty, \beta]$. Suppose without loss of
generality that $\pi_1(\cc^M) \subset [\beta, +\infty)$, so that
$\beta$ is the minimum value attained by $\pi_1$ on $\cc^M$.

Let $\vecx=(\beta, \x_2, \ldots, \x_n) \in \cc^M$, and suppose that
$\rank (\tilde{\H}(\vecx)) = p$. By Proposition \ref{prop:closedness} (assertion (c)),
there exists $\vecy \in \CC^{p+1}$ such that $(\vecx,\vecy) \in
\incidence(H, \u, p)$. Note that since $\rank (\tilde{\H}(\vecx)) = p$, we also deduce that 
$(\vecx,\vecy) \in
\incidencereg(H, \u, p)$.

We claim that there exists $\vecz \in \CC^{2m-p}$ such that $(\vecx,\vecy,\vecz)$
lies on $\reg(\lagrange(\H^M, \u, \v, p))$.

Since $\sfH$ holds, Proposition \ref{prop:dimension} implies that
$\lagrange(\H^M, \u, \v, p)$ satisfies $\sfG$ over $\incidencereg(H^M,
\u, p)$. Also, note that the Jacobian criterion implies that
$\reg(\lagrange(\H^M, \u, \v, p))$ has dimension at most $0$.
We conclude that the point $\vecx \in \cc^M$ lies on the finite set
encoded by the rational parametrization {\sf P} obtained at Step
\ref{step:rec:3} of ${\sf LowRankHankel}$ and we are done. 

It remains to prove our claim, i.e. there exists $\vecz \in \CC^{2m-p}$
such that $(\vecx,\vecy,\vecz)$ lies on $\reg(\lagrange(\H^M, \u, \v, p))$.

Let $\cc' $ be the connected component of $\setV(\H, \u, p)^M \cap
\RR^{n+m(m-r)}$ containing $(\vecx,\vecy)$. We first prove that $\beta =
\pi_1(\vecx,\vecy)$ lies on the boundary of $\pi_1(\cc')$. Indeed, suppose
that there exists $(\widetilde{\vecx},\widetilde{\vecy}) \in \cc'$ such that
$\pi_1(\widetilde{\vecx},\widetilde{\vecy}) < \beta$. Since $\cc'$ is
connected, there exists a continuous semi-algebraic map $\tau \colon
[0,1] \to \cc'$ with $\tau(0) = (\vecx,\vecy)$ and $\tau(1) =
(\widetilde{\vecx},\widetilde{\vecy})$. Let $\varphi: (\vecx, \vecy)\to \vecx$ be the
canonical projection on the $\vecx$-space.

Note that $\varphi \circ \tau$ is also continuous and semi-algebraic
(it is the composition of continuous semi-algebraic maps), with
$(\varphi \circ \tau)(0)=\vecx$, $(\varphi \circ \tau)(1)
=\widetilde{\vecx}$. Since $(\varphi \circ \tau)(\theta) \in \setH_p$ for
all $\theta \in [0,1]$, then $\widetilde{\vecx} \in \cc$. Since
$\pi_1(\widetilde{\vecx}) = \pi_1(\widetilde{\vecx}, \widetilde{\vecy}) <
\fiber$ we obtain a contradiction.  So $\pi_1(\vecx,\vecy)$ lies on the
boundary of $\pi_1(\cc')$.

By the Implicit Function Theorem, and the fact that $\f(\H, \u, p)$
satisfies Property ${\sfG}$, one deduces that $(\vecx,\vecy)$ is a critical
point of the restriction of $\Pi_1: (\x_1, \ldots, \x_n, \y_1, \ldots,
\y_{r+1})\to \x_1$ to $\setV(\H, \u, p)$. 

Since ${\sf rank}(\H^M(\vecx))=p$ by construction, we deduce that
$(\vecx,\vecy)$ is a critical point of the restriction of $\Pi_1$ to
$\incidencereg(\H^M, \u, p)$ and that, by Proposition
\ref{prop:dimension}, there exists $\vecz \in \CC^{2m-p}$ such
that $(\vecx,\vecy,\vecz)$ belongs to the set $\reg(\lagrange(\H^M, \u, \v, p))$, as claimed.
\end{proof}

\section{Degree bounds and complexity} \label{ssec:algo:complexity}

We first remark that the complexity of subroutines ${\sf Union}$,
${\sf Lift}$ and ${\sf ChangeVariables}$ (see \cite[Chap. 10]{SaSc13})
are negligible with respect to the complexity of ${\sf
  ZeroDimSolveMaxRank}$.
Hence, the complexity of ${\sf LowRankHankel}$ is at most $n$ times
the complexity of ${\sf ZeroDimSolveMaxRank}$, which is computed
below.

Let $(\H,r)$ be the input, and let $0 \leq p \leq r$.
We estimate the complexity of ${\sf ZeroDimSolveMaxRank}$ with input
$(\H^M, \u_p, \v_p)$. It depends on the algorithm used to solve
zero-dimensional polynomial systems. We choose the one of
\cite{jeronimo2009deformation} that can be seen as a symbolic homotopy
taking into account the sparsity structure of the system to solve.
More precisely, let $\mathbf{p}\subset \QQ[x_1, \ldots, x_n]$ and
$s\in \QQ[x_1, \ldots, x_n]$ such that the common complex solutions of
polynomials in $\mathbf{p}$ at which $s$ does not vanish is finite.
The algorithm in \cite{jeronimo2009deformation} builds a system
$\mathbf{q}$ that has the same monomial structure as $\mathbf{p}$ has
and defines a finite algebraic set. Next, the homotopy system
$\mathbf{t} = t\mathbf{p}+(1-t)\mathbf{q}$ where $t$ is a new
variable is built. The system $\mathbf{t}$ defines a $1$-dimensional
constructible set over the open set defined by $s\neq 0$ and for
generic values of $t$. Abusing notation, we denote by $Z(\mathbf{t})$
the curve defined as the Zariski closure of this constructible set.

Starting from the solutions
of $\mathbf{q}$ which are encoded with a rational parametrization, the
algorithm builds a rational parametrization for the solutions of
$\mathbf{p}$ which do not cancel $s$. Following
\cite{jeronimo2009deformation}, the algorithm runs in time
$\softO(Ln^{O(1)} \delta \delta')$ where $L$ is the complexity of
evaluating the input, $\delta$ is a bound on the number of isolated
solutions of $\mathbf{p}$ and $\delta'$ is a bound on the degree of
$Z(\mathbf{t})$ defined by $\mathbf{t}$.

Below, we estimate these degrees when the input is a Lagrange system
as the ones we consider.

\noindent
{\bf Degree bounds.} We let $((\tilde{\H}
\,\vecy)',{\u_p}'\vecy-1)$, with $\vecy = (\Y_1, \ldots, \Y_{p+1})'$,
defining $\setV_p(\H,{\u_p})$. Since $\vecy \neq 0$, one can eliminate
w.l.o.g. $\Y_{p+1}$, and the linear form ${\u_p}'\vecy-1$, obtaining a
system $\tilde{\f} \in \QQ[\vecx,\vecy]^{2m-p-1}$.  We recall that if
$\vecx^{(1)}, \ldots, \vecx^{(c)}$ are $c$ groups of variables, and $f
\in \QQ[\vecx^{(1)}, \ldots, \vecx^{(c)}]$, we say that the
multidegree of $f$ is $(d_1, \ldots, d_c)$ if its degree with respect
to the group $\vecx^{(j)}$ is $d_j$, for $j=1, \ldots, c$.

Let $\bell=(\tilde{\f}, \tilde{\g}, \tilde{\h})$ be the corresponding
Lagrange system, where
$$
(\tilde{\g}, \tilde{\h}) = (\tilde{g}_1, \ldots, \tilde{g}_{n-1}, \tilde{h}_1,
\ldots, \tilde{h}_{p}) = \vecz' \jac_1 \tilde{\f}
$$
with $\vecz = [1, \Z_2, \ldots, \Z_{2m-p-1}]$ a non-zero vector of Lagrange
multipliers (we let $\Z_1 = 1$ w.l.o.g.). One obtains that $\bell$ is constituted by
\begin{itemize}
\item
$2m-p-1$ polynomials of multidegree bounded by $(1,1,0)$ with respect to $(\vecx,\vecy,\vecz)$,
\item
$n-1$ polynomials of multidegree bounded by $(0,1,1)$ with respect to $(\vecx,\vecy,\vecz)$,
\item 
$p$ polynomials of multidegree bounded by $(1,0,1)$ with respect to $(\vecx,\vecy,\vecz)$,
\end{itemize}
that is by $n+2m-2$ polynomials in $n+2m-2$ variables.

\begin{lemma} \label{lemma1}
With the above notations, the number of isolated solutions of $\zeroset{\bell}$ is at most
\[
\delta(m,n,p) = \sum_{\ell}\binom{2m-p-1}{n-\ell} \binom{n-1}{2m-2p-2+\ell} \binom{p}{\ell}
\]
where $\ell\in \{\max\{0,n-2m+p+1\}, \ldots, 
\min\{p,n-2m+2p+1\}\}$.
\end{lemma}
\begin{proof}
By \cite[Proposition 11.1]{SaSc13}, this degree is bounded by the multilinear
B\'ezout bound $\delta(m,n,p)$ which is the sum of the coefficients of
\[
(s_\X+s_\Y)^{2m-p-1} (s_\Y+s_\Z)^{n-1} (s_\X+s_\Z)^{p} \in \QQ[s_\X,s_\Y,s_\Z]
\]
modulo $I = \left \langle s_\X^{n+1}, s_\Y^{p+1}, s_\Z^{2m-p-1} \right
\rangle$.  The conclusion comes straightforwardly by technical 
computations.
\end{proof}


With input $\bell$, the homotopy system $\mathbf{t}$ is constituted by
$2m-p-1, n-1$ and $p$ polynomials of multidegree respectively bounded
by $(1,1,0,1), (0,1,1,1)$ and $(1,0,1,1)$ with respect to
$(\vecx,\vecy,\vecz,t)$. 
We prove the following.

\begin{lemma} \label{lemma2}
$\deg \, \zeroset{\mathbf{t}} \in {O}(pn(2m-p) \delta(m,n,p))$.
\end{lemma}
\begin{proof}[of Lemma \ref{lemma2}]
We use Multilinear B\'ezout bounds as in the proof of Lemma \ref{lemma1}.
The degree of $\zeroset{{\bf t}}$ is bounded by the sum of the coefficients
of
\[
(s_\X+s_\Y+s_{t})^{2m-p-1} (s_\Y+s_\Z+s_{t})^{n-1} (s_\X+s_\Z+s_{t})^{p}
\]
modulo $I = \left\langle s_\X^{n+1}, s_\Y^{p+1}, s_\Z^{2m-p-1}, s_{t}^2 \right\rangle
\subset \QQ[s_\X, s_\Y, s_\Z, s_{t}]$. Since the variable $s_t$ can appear
up to power $1$, the previous polynomial is congruent to $P_1+P_2+P_3+P_4$
modulo $I$, with
\begin{itemize}
\item[] $P_1 = (s_\X+s_\Y)^{2m-p-1} (s_\Y+s_\Z)^{n-1} (s_\X+s_\Y)^{p}$
\item[] $P_2 = (2m-p-1) s_t (s_\X+s_\Y)^{2m-p-2} (s_\Y+s_\Z)^{n-1} \,(s_\X+s_\Z)^{p}$,
\item[] $P_3 = (n-1) s_t (s_\Y+s_\Z)^{n-2} (s_\X+s_\Y)^{2m-p-1} (s_\X+s_\Z)^{p}$
\item[] $P_4 = p \, s_t (s_\X+s_\Z)^{p-1} (s_\X+s_\Y)^{2m-p-1} (s_\Y+s_\Z)^{n-1}.$
\end{itemize}
We denote by $\Delta(P_i)$ the contribution of $P_i$ to the previous
sum.

Firstly, observe that $\Delta(P_1) = \delta(m,n,p)$ (compare with the
proof of Lemma \ref{lemma1}). Defining
$\chi_1 = \max\{0,n-2m+p+1\}$ and $\chi_2 = \min\{p,n-2m+2p+1\}$,
one has $\Delta(P_1) = \delta(m,n,p) = \sum_{\ell = \chi_1}^{\chi_2}\bX(\ell)$
with
$\bX(\ell) = \binom{2m-p-1}{n-\ell} \binom{n-1}{2m-2p-2+\ell} \binom{p}{\ell}.$

Write now $P_2 = (2m-p-1)s_t\tilde{P}_2$, with $\tilde{P}_2 \in \QQ[\X,\Y,\Z].$
Let $\Delta(\tilde{P}_2)$ be the contribution of $\tilde{P}_2$, that is the sum
of the coefficients of $\tilde{P}_2$ modulo $I' = \left\langle s_\X^{n+1}, s_\Y^{p+1},
s_\Z^{2m-p-1} \right\rangle$, so that $\Delta(P_2) = (2m-p-1)\Delta(\tilde{P}_2)$. Then
\[
\Delta(\tilde{P}_2) = \sum_{i,j,\ell}\binom{2m-p-2}{i}\binom{n-1}{j}\binom{p}{\ell}
\]
where the sum runs in the set defined by the inequalities
\[
i + \ell \leq n, \,\,\, 2m-p-2-i+j \leq p, \,\,\, n-1-j+p-\ell \leq 2m-p-2.
\]
Now, since $\tilde{P}_2$ is homogeneous of degree $n+2m-3$, only three possible
cases hold:

{\it Case (A)}. $i + \ell=n$, $2m-p-2-i+j =p$ and $n-1-j+p-\ell =2m-p-3$. Here the contribution is
  $\delta_a = \sum_{\ell=\alpha_1}^{\alpha_2} \bY_a(\ell)$ with
  \[
  \bY_a(\ell) = \binom{2m-p-2}{n-\ell}\binom{n-1}{2m-2p-3+\ell}\binom{p}{\ell},
  \]
  and $\alpha_1 = \max\{0,n-2m+p+2\}, \alpha_2 = \min\{p,n-2m+2p+2\}.$
  Suppose first that $\ell$ is an admissible index for $\Delta(P_1)$ and $\delta_a$,
  that is $\max\{\chi_1,\alpha_1\}=\alpha_1 \leq \ell \leq \chi_2=\min\{\chi_2,\alpha_2\}$.
  Then:
  \begin{align*}
  \bY_a(\ell) & \leq \binom{2m-p-1}{n-\ell}\binom{n-1}{2m-2p-3+\ell}\binom{p}{\ell} = \\
              & = \Psi(\ell) \bX(\ell) \qquad \text{with} \, \Psi(\ell) = \frac{2m-2p-2+\ell}{n-(2m-2p-2+\ell)}.
  \end{align*}
  The rational function $\ell \longmapsto \Psi(\ell)$
  is piece-wise monotone (its first derivative is positive), and its unique possible
  pole is $\ell = n-2m+2p+2$. Suppose that this value is a pole for $\Psi(\ell)$.
  This would imply $\alpha_2 = n-2m+2p+2$ and so $\chi_2 = n-2m+2p+1$; since $\ell$
  is admissible for $\Delta(P_1)$, then one would conclude a contradiction. Hence
  the rational function $\Psi(\ell)$ has no poles, its maximum is atteined in
  $\chi_2$ and its value is $\Psi(\chi_2) = n-1$. Hence 
  $\bY_a(\ell) \leq (n-1)\bX(\ell)$.
  Now, we analyse any possible case:
  \begin{enumerate}
  \item[(A1)] $\chi_1 = 0, \alpha_1=0$. This implies $\chi_2=n-2m+2p+1, \alpha_2=n-2m+2p+2$.
    We deduce that 
    \begin{align*}
    \delta_a &= \sum_{\ell = 0}^{\chi_2}\bY_a(\ell) + \bY_a(\alpha_2) \leq (n-1) \sum_{\ell = 0}^{\chi_2}\bX(\ell) + \\
             &+ \bY_a(\alpha_2)  \leq (n-1)\Delta(P_1)+ \bY_a(\alpha_2).
    \end{align*}
    In this case we deduce the bound $\delta_a \leq n  \Delta(P_1)$.
  \item[(A2)] $\chi_1 = 0, \alpha_1=n-2m+p+2$. This implies $\chi_2=n-2m+2p+1, \alpha_2=p$.
    In this case all indices are admissible, and hence we deduce the bound $\delta_a \leq (n-1) \Delta(P_1).$
  \item[(A3)] $\chi_1 = n-2m+p+1$. This implies $\alpha_1=n-2m+p+2$, $\chi_2=p, \alpha_2=p$.
    Also in this case all indices are admissible, and $\delta_a \leq (n-1) \Delta(P_1).$
  \end{enumerate}

{\it Case (B)}. $i + \ell=n$, $2m-p-2-i+j =p-1$ and $n-1-j+p-\ell = 2m-p-2$. Here the contribution is
  $\delta_b=\sum_\ell \bY_b(\ell)$ where
  \[
  \bY_b(\ell) = \binom{2m-p-2}{n-\ell}\binom{n-1}{2m-2p-2+\ell}\binom{p}{\ell}.
  \]
  One gets $\delta_b \leq \Delta(P_1)$ since the sum above is defined in $\max \{0,n\-2m+p+2\} \leq
  \ell \leq \min \{p,n-2m+2p+1\}$, and the inequality $\bY_b(\ell) \leq \bX(\ell)$ holds term-wise.

{\it Case (C)} $i + \ell=n-1$, $2m-p-2-i+j = p$ and $n-1-j+p-\ell = 2m-p-2$. Here the contribution is
  $\delta_c = \sum_{\ell}\bY_c(\ell)$ where
  \[
  \bY_c(\ell) = \binom{2m-p-2}{n-1-\ell}\binom{n-1}{2m-2p-2+\ell}\binom{p}{\ell}.
  \]
  One gets $\delta_c \leq \Delta(P_1)$ since the sum above is defined in $\max\{0,n\-2m+p+1\} \leq
  \ell \leq \min\{p,n-2m+2p+1\}$, and the inequality $\bY_c(\ell) \leq \bX(\ell)$ holds term-wise.
We conclude that $\delta_a \leq n  \Delta(P_1)$, $\delta_b \leq \Delta(P_1)$ and $\delta_c \leq \Delta(P_1)$.
Hence $\Delta(P_2) = (2m-p-1) (\delta_a+\delta_b+\delta_c) \in {O}(n(2m-p)  \Delta(P_1)).$

Analogously to $\Delta(P_2)$, one can conclude that $\Delta(P_3) \in {O}(n(n+2m-p) \Delta(P_1))$
and $\Delta(P_4) \in {O}(pn(n+2m-p) \Delta(P_1)).$
\end{proof}

\noindent{\bf Estimates.}\\
We provide the whole complexity of ${\sf ZeroDimSolveMaxRank}$.
\begin{theorem}
Let $\delta = \delta(m,n,p)$ be given by Lemma \ref{lemma1}. Then
{\sf ZeroDimSolveMaxRank} with input $\bell(\H^M, \u_p, \v_p)$ computes
a rational parametrization 
within
\[
\softO(p(n+2m)^{O(1)}(2m-p) \delta^2),
\]
arithmetic operations over $\QQ$.
\end{theorem}
\begin{proof}
  The polynomial entries of the system $\mathbf{t}$ (as defined in the
  previous section) are cubic polynomials in $n+2m-1$ variables, so
  the cost of their evaluation is in $O((n+2m)^3)$. Applying
  \cite[Theorem 5.2]{jeronimo2009deformation} and bounds given in
  Lemma \ref{lemma1} and \ref{lemma2} yield the claimed complexity
  estimate.
\end{proof}

From Lemma \ref{lemma1}, one deduces that for all $0 \leq p \leq r$,
the maximum number of complex solutions computed by ${\sf
  ZeroDimSolveMaxRank}$ is bounded above by $\delta(m,n,p)$. We deduce
the following result.

\begin{proposition}
  Let $H$ be a $m \times m$, $n-variate$ linear Hankel matrix, and let
  $r \leq m-1$.  The maximum number of complex solutions computed by
  ${\sf LowRankHankel}$ with input $(H,r)$ is
$$
\binom{2m-r-1}{r} + \sum_{k=2m-2r}^{n}\sum_{p=0}^{r} \delta(m,k,p).
$$
where 
$\delta(m,k,p)$ is the bound defined in Lemma \ref{lemma1}.
\end{proposition}
\begin{proof}
  The maximum number of complex solutions computed by {\sf
    ZeroDimSolve} is the degree of $\setV(H, \u, r)$. Using, the
  multilinear B\'ezout bounds, this is bounded by the coefficient of
  the monomial $s_\X^{n}s_\Y^{r}$ in the expression
  $(s_\X+s_\Y)^{2m-r-1}$, that is exactly $\binom{2m-r-1}{r}$.  The
  proof is now straightforward, since {\sf ZeroDimSolveMaxRank} runs
  $r+1$ times at each recursive step of {\sf LowRankHankel}, and since
  the number of variables decreases from $n$ to $2m-2r$.
\end{proof}

\section{Proof of Proposition \ref{prop:dimension}} \label{sec:dimension}

\noindent
We start with a local description of the algebraic sets defined by our
Lagrange systems. This is obtained from a local description of the
system defining $\incidence(\H, \u, p)$. Without loss of generality,
we can assume that $\u=(0, \ldots, 0, 1)$ in the whole section: such a
situation can be retrieved from a linear change of the $\vecy$-variables
that leaves invariant the $\vecx$-variables.


\subsection{Local  equations} \label{ssec:dimlag:local}\label{sssec:dimlag:local:inc}
\label{sssec:dimlag:local:lag}



Let $(\vecx, \vecy)\in \incidencereg(\H, \u, p)$. Then, by definition, there
exists a $p \times p$ minor of $\tilde{\H}(\vecx)$ that is
non-zero. Without loss of generality, we assume that this minor is the
determinant of the upper left $p\times p$ submatrix of
$\tilde{H}$. Hence, consider the following block partition
\begin{equation} \label{partition}
\tilde{\H}(\vecx) = 
\left[
\begin{array}{cc}
  N  & Q \\
  P  & R \\
\end{array}
\right]
\end{equation}
with $N \in \QQ[\vecx]^{p \times p}$, and $Q \in \QQ[\vecx]^{p}$, $P \in
\QQ[\vecx]^{(2m-2p-1) \times p}$, and $R \,\in \,\QQ[\vecx]^{2m-2p-1}$. We
are going to exhibit suitable local descriptions of $\incidencereg(\H,
\u, p)$ over the Zariski open set $O_N\subset \CC^{n+p+1}$ defined by
$\det N \neq 0$; we denote by $\QQ[\vecx,\vecy]_{\det N}$ the local ring of
$\QQ[\vecx, \vecy]$ localized by $\det N$.


\begin{lemma} \label{lemma:local:incidence}
Let $N,Q,P,R$ be as above, and $\u \in \QQ^{p+1}-\{\mathbf{0}\}$. Then there exist
$\{q_{i}\}_{1 \leq i \leq p} \subset \QQ[\vecx]_{\det N}$ and
$\{\tilde{q}_{i}\}_{1 \leq i \leq 2m-2p-1} \subset \QQ[\vecx]_{\det N}$ such that the
constructible set $\incidencereg(\H, \u, p) \cap O_N$ is
defined by the equations
\begin{align*}
\Y_{i} - q_{i}(\vecx) &= 0 \qquad 1 \leq i \leq p \\
\tilde{q}_{i}(\vecx) &= 0 \qquad 1 \leq i \leq 2m-2p-1 \\
\Y_{p+1} - 1 &= 0.
\end{align*}
\end{lemma}
\begin{proof}
Let $c=2m-2p-1$.
The proof follows by the equivalence
\[
\left[ \begin{array}{cc} N & Q \\ P & R \end{array} \right] \vecy = 0
\;
\text{iff}
\;
\left[ \begin{array}{cc} \Id_p & 0 \\ -P & \Id_{c} \end{array} \right]
\left[ \begin{array}{cc} N^{-1} & 0 \\ 0 & \Id_{c} \end{array} \right]
\left[ \begin{array}{cc} N & Q \\ P & R \end{array} \right] \vecy = 0
\]
in the local ring $\QQ[\vecx,\vecy]_{\det N}$, that is if and only if
\[
\left[ \begin{array}{cc} \Id_p & N^{-1}Q \\ 0 & R-PN^{-1}Q \end{array} \right] \vecy = 0
\]
Recall that we have assumed that $\u=(0, \ldots, 0, 1)$; then the
equation $\u\vecy=1$ is $\Y_{p+1}=1$. Denoting by $q_{i}$ and
$\tilde{q}_{i}$ respectively the entries of vectors $-N^{-1}Q$ and
$-(R-PN^{-1}Q)$ ends the proof.
\end{proof}

The above local system is denoted by $\tilde{\f} \in \QQ[\vecx,\vecy]_{\det N}^{2m-p}$.
The Jacobian matrix of this polynomial system is
\[
\jac\tilde{\f} =
\left[
\begin{array}{cc}
\begin{array}{c} \jac_x\tilde{\q} \\ \star \end{array}
&
\begin{array}{c} 0 \\ \Id_{p+1} \end{array}
\end{array}
\right]
\]
with $\tilde{\q} = (\tilde{q}_{1}(\vecx), \ldots,
\tilde{q}_{2m-2p-1}(\vecx))$. Its kernel defines the tangent space to
$\incidencereg(\H, \u, p)\cap O_N$.  Let $\w=(\w_1, \ldots, \w_n) \in
\CC^n$ be a row vector; we denote by $\pi_\w$ the projection
$\pi_\w(\vecx,\vecy) = \w_1\X_1 + \cdots + \w_n\X_n$.  
Given a row vector $\v \in \CC^{2m-p+1}$, we denote by ${\sf wlagrange}(\tildef,
\v)$ the following polynomial system
\begin{equation} \label{local-lag}
\tildef, \,\,\,\, (\tilde{\g}, \tilde{\h}) = [\Z_1, \ldots, \Z_{2m-p}, \Z_{2m-p+1}]
\left[
\begin{array}{c}
\jac \tildef \\
\begin{array}{cc}
\w & 0
\end{array}
\end{array}
\right], \,\,\,\,
\v'\vecz-1.
\end{equation}
For all $0 \leq p \leq r$, this polynomial system contains $n+2m+2$
polynomials and $n+2m+2$ variables. We denote by ${\sf L}_p(\tildef, \v,
\w)$ the set of its solutions whose projection on the $(\vecx, \vecy)$-space
lies in $O_N$.

Finally, we denote by ${\sf wlagrange}({\f}, \v)$ the polynomial
system obtained when replacing $\tildef$ above with $\f=\f(\H, \u,
p)$. Similarly, its solution set is denoted by ${\sf L}_p({\f}, \v,
\w)$.



\subsection{Intermediate result} \label{ssec:dimlag:intlemma}



\begin{lemma} \label{lemma:intermediate}
Let $\zarH \subset \CC^{(2m-r)(n+1)}$ be the non-empty Zariski open set
defined by Proposition \ref{prop:regularity}, $\H \in \zarH$ 
and $0 \leq p \leq r$.
There exist non-empty Zariski open sets $\zarV \subset \CC^{2m-p}$ and
$\zarW \subset \CC^n$ such that if $\v \in \zarV$ and $\w \in
\zarW$, the following holds:
\begin{itemize}
\item[(a)] the set $\L_p(\f, \v, \w)=\L(\f, \v, \w) \cap
  \{(\vecx,\vecy,\vecz) \mid {\rank}\,\tilde{H}(\vecx)=p\}$ is finite and the
  Jacobian matrix of ${\sf wlagrange}({\f}, \v)$ has maximal rank at
  any point of $\L_p(\f, \v, \w)$;
\item[(b)] the projection of $\L_p(\f, \v, \w)$ in the
  $(\vecx,\vecy)$-space contains the critical points of the restriction of
  $\pi_\w$ restricted to $\incidencereg(\H, \u, p)$.
\end{itemize}
\end{lemma}
\proof
We start with Assertion (a).

The statement to prove holds over $\incidencereg(\H, \u, p)$; hence it
is enough to prove it on any open set at which one $p\times p$ minor
of $\tilde \H$ is non-zero. Hence, we assume that the determinant of
the upper left $p\times p$ submatrix $N$ of $\tilde \H$ is non-zero;
$O_N\subset \CC^{n+p+1}$ is the open set defined by $\det\, N \neq 0$,
and we reuse the notation introduced in this section. We prove
that there exist non-empty Zariski open sets $\zarV'_N\subset \CC^{2m-p}$
and $\zarW_N \subset \CC^{n}$ such that for $\v \in \zarV'_N$ and $\w \in
\zarW_N$, $\L_p(\tildef, \v, \w)$ is finite and that the Jacobian matrix
associated to ${\sf wlagrange}(\tildef, \v)$ has maximal rank at any
point of $\L_p(\tildef, \v, \w)$. The Lemma follows straightforwardly
by defining $\zarV'$ (resp. $\zarW$) as the intersection of $\zarV'_N$
(resp. $\zarW_N$) where $N$ varies in the set of $p \times p$ minors
of $\tilde{H}(\vecx)$.

Equations $\tilde{\h}$ yield $\Z_{j}=0$ for $j=2m-2p, \ldots, 2m-p$,
and can be eliminated together with their
$\vecz$ variables from the Lagrange system ${\sf wlagrange}(\tildef,
\v)$. It remains $\vecz$-variables $\Z_1, \ldots, \Z_{2m-2p-1}, \Z_{2m-p+1}$;
we denote by $\Omega \subset \CC^{2m-2p}$ the Zariski open set where they don't
vanish simultaneously.

Now, consider the map
\[
  \begin{array}{lrcc}
  q : &  O_N \times \Omega \times \CC^{n} & \longrightarrow & \CC^{n+2m-p} \\
            &  (\vecx,\vecy,\vecz,\w) & \longmapsto & (\tildef, \tilde{\g})
  \end{array}
\]
and, for $\w \in \CC^n$, its section map $q_{\w}(\vecx,\vecy,\vecz) = q(\vecx,\vecy,\vecz,\w)$.
We consider $\tilde{\v} \in \CC^{2m-p}$ and we denote by $\tilde{\vecz}$
the remaining $\vecz-$variables, as above. Hence we define
\[
  \begin{array}{lrcc}
  Q : &  O_N \times \Omega \times \CC^{n} \times \CC^{2m-2p} & \longrightarrow & \CC^{n+2m-p+1} \\
            &  (\vecx, \vecy, \vecz, \w, \tilde{\v}) & \longmapsto & (\tildef, \tilde{\g}, \tilde{\v}'\vecz-1)
  \end{array}
\]
and its section map $Q_{\w,\tilde{\v}}(\vecx,\vecy,\vecz) = q(\vecx,\vecy,\vecz,\w,\tilde{\v})$.
We claim that $\mathbf{0} \in \CC^{n+2m-p}$ (resp. $\mathbf{0} \in \CC^{n+2m-p+1}$) is a
regular value for $q$ (resp. $Q$). Hence we deduce, by Thom's Weak Transversality
Theorem, that there exist non-empty Zariski open sets $\zarW_N \subset \CC^n$ and
$\tilde{\zarV}_N \subset \CC^{2m-2p}$ such that if $\w \in \zarW_N$ and $\tilde{\v} \in
\tilde{\zarV}_N$, then $\mathbf{0}$ is a regular value for $q_{\w}$ and $Q_{\w,\tilde{\v}}$.

We prove now this claim. Recall that since $H \in \zarH$, the Jacobian matrix
$\jac_{\vecx,\vecy} \tilde{\f}$ has maximal rank at any point $(\vecx,\vecy) \in \zeroset{\tilde{\f}}$.
Let $(\vecx,\vecy,\vecz,\w) \in q^{-1}(\bf0)$ (resp. $(\vecx, \vecy, \vecz, \w, \tilde{\v}) \in Q^{-1}(\bf0)$).
Hence $(\vecx,\vecy) \in \zeroset{\tilde{\f}}$. We isolate the square submatrix of
$\jac q (\vecx,\vecy,\vecz,\w)$ obtained by selecting all its rows and  
\begin{itemize}
\item the columns corresponding to derivatives of $\vecx, \vecy$ yielding a
  non-singular submatrix of $\jac_{\vecx,\vecy} \tilde{\f}(\vecx,\vecy)$;
\item the columns corresponding to the derivatives w.r.t. $\w_1,
  \ldots, \w_n$, hence this yields a block of zeros when applied to
  the lines corresponding to $\tilde{\f}$ and the block $\Id_n$ when
  applied to $\tilde{\g}$.
\end{itemize}
For the map $Q$, we consider the same blocks as above. Moreover, since
$(\vecx,\vecy,\vecz,\w,\tilde{\v}) \in Q^{-1}(\bf0)$ verifies $\tilde{\v}'\vecz-1=0$,
there exists $\ell$ such that $\z_\ell \neq 0$. Hence, we add the derivative
of the polynomial $\tilde{\v}'\vecz-1$ w.r.t. $\tilde{\v}_\ell$, which
is $\z_\ell \neq 0$. The claim is proved.

Note that $q^{-1}_\w(\mathbf{0})$ is defined by $n+2m-p$ polynomials
involving $n+2m-p+1$ variables. We deduce that for $\w \in \zarW_N$,
$q_\w^{-1}(\mathbf{0})$
is either empty or it is equidimensional and has dimension $1$. Using
the homogeneity in the $\vecz$-variables and the Theorem on the Dimension of
Fibers \cite[Sect. 6.3, Theorem 7]{Shafarevich77}, we deduce that the projection on the $(\vecx, \vecy)$-space of
$q_\w^{-1}(\mathbf{0})$ has dimension $\leq 0$.
We also deduce that for $\w \in \zarW_N$ and $\tilde{\v} \in \tilde{\zarV}_N$,
$Q_{\w,\tilde{\v}}^{-1}(\bf0)$ is either empty or finite.

Hence, the points of $Q^{-1}_{\v, \w}(\mathbf{0})$ are in bijection
with those in $\L(\tildef, \v, \w)$ forgetting their $0$-coordinates
corresponding to $\Z_j=0$.
We define $\zarV'_N = \tilde{\zarV}_N \times \CC^{p} \subset \CC^{2m-2p}$.
We deduce straightforwardly that for $\v \in \zarV'_N$ and $\w \in \zarW_N$,
the Jacobian matrix of ${\sf wlagrange}(\tildef, \v)$ has
maximal rank at any point of $\L_p(\tildef, \v, \w)$. By the Jacobian
criterion, this also implies that the set $\L_p(\tildef, \v, \w)$ is
finite as requested.

We prove now Assertion (b).

Let $\zarW \subset \CC^n$ and $\zarV' \subset \CC^{2m-p}$ be the non-empty
Zariski open sets defined in the proof of Assertion (a). For $\w \in \zarW$
and $\v \in \zarV'$, the projection of $\L_p(\tildef, \v, \w)$ on the
$(\vecx,\vecy)-$space is finite.
Since $H \in \zarH$, $\incidencereg(\H, \u, p)$ is smooth and
equidimensional.

Since we work on $\incidencereg(\H, \u, p)$, one of the $p \times p$
minors of $\tilde{H}(\vecx)$ is non-zero. Hence, suppose to work in
$O_N \cap \incidencereg(\H, \u, p)$ where $O_N \subset \CC^{n+p+1}$
has been defined in the proof of Assertion (a). Remark that
\[
\crit(\pi_\w, \incidencereg(\H, \u, p)) \, = \, \bigcup_N \, \crit(\pi_\w, O_N \cap \incidencereg(\H, \u, p))
\]
where $N$ runs over the set of $p \times p$ minors of $\tilde{H}(\vecx)$.
We prove below that there exists a non-empty Zariski open set
$\zarV \subset \CC^{2m-p}$ such that if $\v \in \zarV$, for all $N$
and for $\w \in \zarW$, the set $\crit(\pi_\w, O_N \cap \incidencereg(\H, \u, p))$
is finite and contained in the projection of $\L_p(\f, \v, \w)$. This
straightforwardly implies that the same holds for $\crit(\pi_\w, \incidencereg(\H, \u, p))$.

Suppose w.l.o.g. that $N$ is the upper left $p \times p$ minor of $\tilde{H}(\vecx)$.
We use the notation $\tilde{\f}, \tilde{\g}, \tilde{\h}$ as above. Hence,
the set $\crit(\pi_\w, O_N \cap \incidencereg(\H, \u, p))$ is the image by the
projection $\pi_{\vecx,\vecy}$ over the $(\vecx,\vecy)-$space, of the
constructible set defined by $\tildef, \tilde{\g}, \tilde{\h}$ and $\vecz
\neq 0$. We previously proved that, if $\w \in \zarW_N$, $q^{-1}(\bf0)$ is either empty
or equidimensional of dimension $1$. Hence, the constructible set defined by
$\tildef, \tilde{\g}, \tilde{\h}$ and $\vecz \neq 0$, which is isomorphic
to $q^{-1}(\bf0)$, is either empty or equidimensional of dimension $1$.

Moreover, for any $(\vecx,\vecy) \in \crit(\pi_\w, O_N \cap \incidencereg(\H, \u, p))$,
$\pi_{\vecx,\vecy}^{-1}(\vecx,\vecy)$ has dimension 1, by the homogeneity
of polynomials w.r.t. variables $\vecz$. By the Theorem on the Dimension
of Fibers \cite[Sect. 6.3, Theorem 7]{Shafarevich77}, we deduce that
$\crit(\pi_\w, O_N \cap \incidencereg(\H, \u, p))$ is finite.






For $(\vecx,\vecy) \in \crit(\pi_\w, O_N \cap \incidencereg(\H, \u, p))$, let
$\zarV_{(\vecx,\vecy),N} \subset \CC^{2m-p}$ be the non-empty Zariski open
set such that if $\v \in \zarV_{(\vecx,\vecy),N}$ the hyperplane
$\v'\vecz-1=0$ intersects transversely $\pi_{\vecx,\vecy}^{-1}(\vecx,\vecy)$.
Recall that $\zarV'_N \subset \CC^{2m-p}$ has been defined in the proof of
Assertion (a). Define
$$
\zarV_N = \zarV'_N \cap \bigcap_{(\vecx,\vecy)} \zarV_{(\vecx,\vecy),N}
$$
and $\zarV = \bigcap_N \zarV_N$. This concludes the proof, since $\zarV$
is a finite intersection of non-empty Zariski open sets.
\foorp

\subsection{Conclusion} \label{ssec:dimlag:proof}
We denote by $\zarM_1 \subset \GL(n,\CC)$ the set of non-singular matrices
$M$ such that the first row $\w$ of $M^{-1}$ lies in the set $\zarW$
given in Lemma \ref{lemma:intermediate}: this set is non-empty and Zariski
open since the entries of $M^{-1}$ are rational functions of the entries of $M$.
Let $\zarV \subset \CC^{2m-p}$ be the non-empty Zariski open set given by Lemma
\ref{lemma:intermediate} and let $\v \in \zarV$.
Let $\e_1$ be the row vector $(1, 0, \ldots, 0) \in \QQ^n$ and for all $M \in \GL(n,\CC)$, let
\[ 
\tilde{M} =
\left[
\begin{array}{cc}
  M & {0} \\
{0} & \Id_m \\
\end{array}
\right]. 
\]
Remark that for any $M \in \zarM_1$ the following identity holds:
$$
\left[\begin{array}{c}
  \jac \f(\H^M, \u, p) \\
  \e_1 \quad 0\; \cdots \; 0\\  
\end{array}\right] = \left[\begin{array}{cc}
  \jac \f(\H, \u, p) \\
\w \quad 0 \; \cdots \; 0\\
\end{array}\right]\tilde{M}.
$$
We conclude that the set of solutions of the system 
\begin{equation}
  \label{eq:dim:1}
  \left(\f(\H, \u, p), \quad
  \vecz'
\left[\begin{array}{c}
  \jac\f(\H, \u, p) \\
  \w \quad 0\; \cdots \; 0\\  
\end{array}\right],
\quad \v'\vecz-1 \right)
\end{equation}
is the image by the map $(\vecx,\vecy) \mapsto \tilde{M}^{-1}(\vecx,\vecy)$
of the set ${S}$ of solutions of the system
\begin{equation}
  \label{eq:dim:2}
  \left(\f(\H, \u, p), \quad
 \vecz'\left[\begin{array}{c}
      \jac\f(\H, \u, p) \\
      \e_1 \quad 0\; \cdots \; 0\\  
\end{array}\right], \quad \v'\vecz-1 \right).
\end{equation}
Now, let $\varphi$ be the projection that eliminates the last
coordinate $\z_{2m-p+1}$. Remark that $\varphi(S)
= {\sf L}_p(\f^M, \v, \e_1)$
.

Now, applying Lemma \ref{lemma:intermediate} ends the proof. 

\foorp

\section{Proof of Proposition \ref{prop:closedness}}\label{sec:closedness}

The proof of Proposition \ref{prop:closedness} relies on results of
\cite[Section 5]{HNS2014} 
and of \cite{SaSc03}.  We use the
same notation as in \cite[Section 5]{HNS2014}, and we recall them below.
\paragraph*{Notations} For $\setZ \subset \CC^n$ of dimension $d$, we denote by
$\Omega_i(\setZ)$ its $i-$equidimensional component, $i=0, \ldots, d$. We denote by
$\scS(\setZ)$ the union of:
\begin{itemize}
\item $\Omega_0(\setZ) \cup \cdots \cup \Omega_{d-1}(\setZ)$
\item the set $\sing(\Omega_d(\setZ))$ of singular points of $\Omega_d(\setZ)$.
\end{itemize}
Let $\pi_i$ be the map $(\x_1, \ldots, \x_n) \to (\x_1, \ldots, \x_i)$.
We denote by $\scC(\pi_i, \setZ)$ the Zariski closure of the union of the
following sets:
\begin{itemize}
\item $\Omega_0(\setZ) \cup \cdots \cup \Omega_{i-1}(\setZ)$;
\item the union for $r \geq i$ of the sets $\crit(\pi_i, \reg(\Omega_r(\setZ)))$.
\end{itemize}
For $M \in \GL(n,\CC)$ and $\setZ$ as above, we define the collection
of algebraic sets $\{\mathcal{O}_i(\setZ^M)\}_{0 \leq i \leq d}$ as follows:
\begin{itemize}
\item ${\mathcal O}_d(\setZ^M)=\setZ^M$;
\item ${\mathcal O}_i(\setZ^M)=\scS({\mathcal O}_{i+1}(\setZ^M))
\cup \scC(\pi_{i+1},  {\mathcal O}_{i+1}(\setZ^M)) \cup \\
\cup \scC(\pi_{i+1},\setZ^M)$ for $i=0, \ldots, d-1$.
\end{itemize}
We finally recall the two following properties: 

{\it Property $\sfP(\setZ)$.} Let ${\setZ} \subset \CC^n$ be
an algebraic set of dimension $d$. We say that $M \in \GL(n,\CC)$ satisfies
$\sfP(\setZ)$ when for all $i = 0, 1, \ldots, d$:
\begin{enumerate}
\item 
${\mathcal O}_i({\setZ}^M)$ has dimension $\leq i$ and
\item 
${\mathcal O}_i({\setZ}^M)$ is in Noether position with respect to $\x_1, \ldots, \x_i$.
\end{enumerate}

{\it Property ${\sf Q}$.} We say that an algebraic set $\setZ$ of dimension $d$
satisfies ${\sf Q}_i(\setZ)$ (for a given $1 \leq i \leq d$) if for any connected
component $\cc$ of ${\setZ}\cap \RR^n$ the boundary of $\pi_i(\cc)$ is contained
in $\pi_i({\mathcal O}_{i-1}(\setZ) \cap {\cc})$. We say that $\setZ$ satisfies
$\sfQ$ if it satisfies $\sfQ_1, \ldots, \sfQ_d$.

Let $\setZ\subset \CC^n$ be an algebraic set of dimension $d$. By
\cite[Proposition 15]{HNS2014}, there exists a non-empty Zariski open
set $\mathscr{M}\subset \GL(n,\CC)$ such that for $M\in
\mathscr{M}\cap \GL(n,\QQ)$ Property ${\sf P}(\setZ)$ holds. Moreover,
if $M \in \GL(n, \QQ)$ satisfies ${\sf P}(\setZ)$, then ${\sf
  Q}_i(\setZ^M)$ holds for $i=1, \ldots, d$ \cite[Proposition
16]{HNS2014}.
We use these results in the following proof of Proposition \ref{prop:closedness}.

\proof We start with assertion (a).  Let $\zarM_2 \subset
  \GL(n,\CC)$ be the non-em\-pty Zariski open set of \cite[Proposition
  17]{HNS2014} for $\setZ = \setH_p$: for $M \in \zarM_2$, $M$
  satisfies ${\sfP}(\setH_p)$.  Remark that the  connected components of
  $\setH_p \cap \RR^n$ and are in bijection with those of $\setH^M_p \cap \RR^n$ (given by
  $\cc \leftrightarrow \cc^M$). Let $\cc^M \subset \setH^M_p \cap
  \RR^n$ be a connected component of $\setH^M_p \cap \RR^n$. Let
  $\pi_1$ be the projection on the first variable $\pi_1 \colon
  \RR^{n} \to \RR$, and consider its restriction to $ \setH^M_r \cap
  \RR^n$.  Since $M \in \zarM_2$, by \cite[Proposition 16]{HNS2014}
  the boundary of $\pi_1(\cc^M)$ is included in $\pi_1({\mathcal
    O}_0(\setH^M_p) \cap \cc^M)$ and in particular in
  $\pi_1(\cc^M)$. Hence $\pi_1(\cc^M)$ is closed.

We prove now assertion (b). 
Let $M \in \zarM_2$, $\cc$ a connected component of $\setH_p \cap \RR^n$ and
$\fiber \in \RR$ be in the boundary of $\pi_1(\cc^M)$. By \cite[Lemma 19]{HNS2014}
$\pi_1^{-1}(\fiber) \cap \cc^M$ is finite.

We claim that, up to genericity
assumptions on $\u \in \QQ^{p+1}$, for $\vecx \in \pi_1^{-1}(\fiber) \cap \cc^M$,
the linear system $\vecy \mapsto \f(\H^M, \u, p)$ has at least one solution. 
We deduce that
there exists a non-empty Zariski open set $\zarU_{\cc,\vecx} \subset
\CC^{p+1}$ such that if $\u \in \zarU_{\cc,\vecx} \cap \QQ^{p+1}$, there exists $\vecy \in
\QQ^{p+1}$ such that $(\vecx,\vecy) \in \setV(\H^M, \u, p)$. One concludes by taking
\[
\zarU =  \bigcap_{\cc \subset \setH_p \cap \RR^n} \bigcap_{\vecx \in \pi_1^{-1}(\fiber) \cap \cc^M} \zarU_{\cc,\vecx},
\]
which is non-empty and Zariski open since:
\begin{itemize}
\item the collection $\{\cc \subset \setH_p \cap \RR^n \, \text{connected component}\}$
  is finite;
\item the set $\pi_1^{-1}(\fiber) \cap \cc^M$ is finite.
\end{itemize}
It remains to prove the claim we made. For $\vecx \in \pi_1^{-1}(\fiber) \cap \cc^M$, the matrix
$\tilde{\H}(\vecx)$ is rank defective, and let $p' \leq p$ be its rank.
The linear system
\[
\left[ \begin{array}{c} \tilde{\H}(\vecx) \\ \u \end{array} \right] \cdot \vecy =
\left[ \begin{array}{c} {\bf 0} \\ 1 \end{array} \right]
\]
has a solution if and only if
\[
\rank \left[ \begin{array}{c} \tilde{\H}(\vecx) \\ \u \end{array} \right] =
\rank \left[ \begin{array}{c} \tilde{\H}(\vecx) \\ \u \end{array}
\begin{array}{c} {\bf 0} \\ 1 \end{array} \right],
\]
and the rank of the second matrix is $p'+1$. Denoting by $\zarU_{\cc,\vecx}
\subset \CC^{p+1}$ the complement in $\CC^{p+1}$ of the $p'-$dimensional
linear space spanned by the rows of $\tilde{\H}(\vecx)$, proves the claim and
concludes the proof.\foorp

\section{Experiments} \label{sec:exper}

The algorithm {\sf LowRankHankel} has been implemented under
\textsc{Ma\-ple}.  We use the \textsc{FGb} \cite{faugere2010fgb}
library implemented by J.-C.  Faugère for solving solving
zero-dimensional polynomial systems using Gr\"obner bases. In
particular, we used the new implementation of \cite{FM11} for
computing rational parametrizations. Our implementation checks the
genericity assumptions on the input.

We test the algorithm with input $m \times m$ linear Hankel matrices
$\H(\vecx)=\H_0+\X_1\H_1+\ldots+\X_n\H_n$, where the entries of $\H_0,
\ldots, \H_n$ are random rational numbers, and an integer $0 \leq r
\leq m-1$. None of the implementations of Cylindrical Algebraic
Decomposition solved our examples involving more that $3$ variables.
Also, on all our examples, we found that the Lagrange systems define
finite algebraic sets. 

We compare the practical behavior of ${\sf LowRankHankel}$ with the
performance of the library {\sc RAGlib}, implemented by the third
author (see \cite{raglib}). Its function ${\sf PointsPerComponents}$,
with input the list of $(r+1)-$minors of $\H(\vecx)$, returns one
point per connected component of the real counterpart of the algebraic
set $\setH_r$, that is it solves the problem presented in this
paper. It also uses critical point methods. The symbol $\infty$ means
that no result has been obtained after $24$ hours. The symbol matbig
means that the standard limitation in $\textsc{FGb}$ to the size of
matrices for Gr\"obner bases computations has been reached.

We report on timings (given in seconds) of the two implementations in
the next table. The column ${\sf New}$ corresponds to timings of ${\sf
  LowRankHankel}$.  Both computations have been done on an Intel(R)
Xeon(R) CPU $E7540$ $@2.00 {\rm GHz}$ 256 Gb of RAM.
We remark that $\textsc{RAGlib}$ is competetitive for problems of small
size (e.g. $m=3$) but when the size increases ${\sf LowRankHankel}$
performs much better, especially when the determinantal variety has
not co-dimension $1$. It can tackle problems that are out reach of
$\textsc{RAGlib}$. Note that for fixed $r$, the algorithm seems to have a
behaviour that is polynomial in $nm$ (this is particularly visible
when $m$ is fixed, e.g. to $5$).

{\tiny
\begin{table}
  \centering
  \begin{tabular}{|c|c|c||c|c|}
    \hline
    $(m,r,n)$ & {\sf RAGlib} & {\sf New} & {\sf TotalDeg} & {\sf MaxDeg}\\
    \hline
    \hline
    $(3,2,2)$ & 0.3 & 5   & 9  & 6\\
    $(3,2,3)$ & 0.6 & 10  & 21 & 12\\
    $(3,2,4)$ & 2 & 13  & 33 & 12\\
    $(3,2,5)$ & 7 & 20  & 39 & 12\\
    $(3,2,6)$ & 13 & 21  & 39 & 12\\
    $(3,2,7)$ & 20 & 21  & 39 & 12\\
    $(3,2,8)$ & 53 & 21  & 39 & 12\\
\hline
\hline
    $(4,2,3)$ & 2 & 2.5 & 10 & 10\\
    $(4,2,4)$ & 43 & 6.5 & 40 & 30\\
    $(4,2,5)$ & 56575 & 18  & 88 & 48\\
    $(4,2,6)$ & $\infty$ & 35  & 128 & 48\\
    $(4,2,7)$ & $\infty$ & 46  & 143 & 48\\
    $(4,2,8)$ & $\infty$ & 74  & 143 & 48\\
\hline
\hline
    $(4,3,2)$ & 0.3 & 8 & 16 & 12\\
    $(4,3,3)$ & 3 & 11 & 36 & 52\\
    $(4,3,4)$ & 54 & 31  & 120 & 68\\
    $(4,3,5)$ & 341 & 112  & 204 & 84\\
    $(4,3,6)$ & 480 & 215  & 264 & 84\\
    $(4,3,7)$ & 528 & 324  & 264 & 84\\
    $(4,3,8)$ & 2638 & 375  & 264 & 84\\
\hline
\hline
    $(5,2,5)$ & 25 & 4  & 21 & 21 \\
    $(5,2,6)$ & 31176 & 21  & 91 & 70\\
    $(5,2,7)$ & $\infty$ & 135  & 199 & 108\\
    $(5,2,8)$ & $\infty$ & 642  & 283 & 108\\
    $(5,2,9)$ & $\infty$ & 950  & 311 & 108\\
    $(5,2,10)$ & $\infty$ & 1106  & 311 & 108\\
\hline
\hline
    $(5,3,3)$ & 2 & 2 & 20 & 20\\
    $(5,3,4)$ & 202 & 18 & 110 & 90\\
    $(5,3,5)$ & $\infty$ & 583  & 338 &228\\
    $(5,3,6)$ & $\infty$ & 6544  & 698 & 360\\
    $(5,3,7)$ & $\infty$ & 28081  & 1058 & 360\\
    $(5,3,8)$ & $\infty$ & $\infty$  & - & - \\
    \hline
  \end{tabular}
  \label{tab:time}
\caption{Timings and degrees}
\end{table}
\begin{table}
  \centering
  \begin{tabular}{|c|c|c||c|c|}
    \hline
    $(m,r,n)$ & {\sf RAGlib} & {\sf New} & {\sf TotalDeg} & {\sf MaxDeg}\\
    \hline
    \hline
\hline
\hline
    $(5,4,2)$ & 1 & 5 & 25 & 20\\
    $(5,4,3)$ & 48 & 30 & 105 & 80\\
    $(5,4,4)$ & 8713 & 885 & 325 & 220\\
    $(5,4,5)$ & $\infty$ & 15537 & 755 & 430\\
    $(5,4,6)$ & $\infty$ & 77962 & 1335 & 580\\
\hline
\hline
    $(6,2,7)$ & $\infty$ & 6  & 36 & 36 \\
    $(6,2,8)$ & $\infty$ & matbig  & - & - \\

\hline
\hline
    $(6,3,5)$ & $\infty$ & 10  & 56 & 56 \\
    $(6,3,6)$ & $\infty$ & 809  & 336 & 280 \\
    $(6,3,7)$ & $\infty$ & 49684  & 1032 & 696 \\
    $(6,3,8)$ & $\infty$ & matbig  & - & - \\

\hline
\hline
    $(6,4,3)$ & 3 & 5  & 35 & 35 \\
    $(6,4,4)$ & $\infty$ & 269  & 245 & 210 \\
    $(6,4,5)$ & $\infty$ & 30660  & 973 & 728 \\
    $(6,4,6)$ & $\infty$ & $\infty$  & - & - \\

\hline
\hline
    $(6,5,2)$ & 1 & 9  &36  & 30 \\
    $(6,5,3)$ & 915 & 356  & 186 & 150 \\
    $(6,5,4)$ & $\infty$ & 20310  & 726 & 540 \\
    $(6,5,5)$ & $\infty$ & $\infty$ & - & - \\
    \hline
  \end{tabular}
  \label{tab:time2}
\caption{Timings and degrees (continued)}
\end{table}
}


Finally, we report in column ${\sf TotalDeg}$ the degree of the
rational parametrization obtained as output of the algorithm, that is
the number of its complex solutions. We observe that this value is
definitely constant when $m,r$ are fixed and $n$ grows, as for the
maximum degree (column ${\sf MaxDeg}$) appearing during the recursive
calls.

The same holds for the multilinear bound given in Section
\ref{ssec:algo:complexity} for the total number of complex solutions.

\end{document}